\documentclass[a4paper,11pt]{article}     
\usepackage{graphicx}
\usepackage{amsmath,amssymb,amsfonts,bbm,mathrsfs,stmaryrd,mathtools}

\usepackage{amsthm}
\usepackage[caption=false,font=footnotesize]{subfig}
\usepackage{xcolor}
\usepackage{xspace}
\usepackage[inline]{enumitem}
\usepackage{bm}
\usepackage{hyperref}

\usepackage[disable]{todonotes}

\newtheorem{theorem}{Theorem}
\newtheorem{lemma}[theorem]{Lemma}
\newtheorem{proposition}[theorem]{Proposition}
\newtheorem{definition}[theorem]{Definition}

\newtheoremstyle{named}{}{}{\itshape}{}{\bfseries}{.}{.5em}{Restatement of #1 \thmnote{#3 }}
\theoremstyle{named}

\newcommand{\paddingsymbol}{\square}

\newcommand{\paddingtuple}[1]{\paddingsymbol^{\cartesianproduct #1}}

\newcommand{\alphabet}{\ensuremath{\Sigma}}
\newcommand{\aalphabet}{\alphabet_1}
\newcommand{\balphabet}{\alphabet_2}

\newcommand{\nbits}{\ensuremath{k}}

\newcommand{\alayer}{\ensuremath{B}}
\newcommand{\blayer}{\ensuremath{\alayer^{\prime}}}
\newcommand{\layer}{\alayer}
\newcommand{\odd}{\ensuremath{D}}
\newcommand{\aodd}{\ensuremath{\odd}}
\newcommand{\bodd}{\ensuremath{\odd^{\prime}}}
\newcommand{\layertransitions}{\ensuremath{T}}
\newcommand{\layerrightfrontier}{\ensuremath{r}}
\newcommand{\layerleftfrontier}{\ensuremath{\ell}}
\newcommand{\layerinitialstates}{\ensuremath{I}}
\newcommand{\layerfinalstates}{\ensuremath{F}}
\newcommand{\layerinitialflag}{\ensuremath{\iota}}
\newcommand{\layerfinalflag}{\ensuremath{\phi}}

\newcommand{\indexset}{\ensuremath{J}}
\newcommand{\lang}{\ensuremath{L}}
\newcommand{\oddlang}[1]{\ensuremath{\mathcal{L}(#1)}}
\newcommand{\automatonlang}[1]{\ensuremath{\bm{\mathcal{L}}(#1)}}
\newcommand{\finiteautomaton}{\ensuremath{\mathcal{F}}}

\newcommand{\automatonstates}{\ensuremath{Q}}
\newcommand{\automatontransitions}{\ensuremath{T}}

\newcommand{\automatoninitialstates}{\ensuremath{I}}
\newcommand{\automatonfinalstates}{\ensuremath{F}}
\newcommand{\N}{\ensuremath{\mathbb{N}}}
\newcommand{\pN}{\ensuremath{\N_+}}
\newcommand{\Q}{\ensuremath{\mathbb{Q}}}

\newcommand{\relation}{\ensuremath{R}}
\newcommand{\arelation}{\relation}
\newcommand{\brelation}{\ensuremath{\relation^{\prime}}}

\newcommand{\relationsymbol}{\ensuremath{\mathsf{R}}}
\newcommand{\structuredomain}{\ensuremath{\relationsymbol_{0}}}

\newcommand{\layeralphabet}[2]{\ensuremath{\mathcal{B}(#1, #2)}}
\newcommand{\layervocabularyalphabet}[3]{\ensuremath{\mathcal{B}(#1, #2, #3)}}

\newcommand{\projection}{\mathsf{proj}}
\newcommand{\identify}{\mathsf{ident}}
\newcommand{\perm}{\mathsf{perm}}
\newcommand{\fold}{\mathsf{fold}}
\newcommand{\unfold}{\mathsf{unfold}}
\newcommand{\directsum}{\oplus}

\newcommand{\width}{\ensuremath{w}}
\newcommand{\vocabulary}{\ensuremath{\tau}}
\newcommand{\arity}{\ensuremath{\mathfrak{a}}}
\newcommand{\arityrelation}{\ensuremath{\arity}}
\newcommand{\fol}[1]{\ensuremath{\mathrm{FO}\set{#1}}}
\newcommand{\msol}[1]{\ensuremath{\mathrm{MSO}\set{#1}}}

\newcommand{\amap}{\ensuremath{\alpha}}

\newcommand{\cmap}{\ensuremath{\beta}}

\newcommand{\permutation}{\ensuremath{\pi}}
\newcommand{\nrelations}{\ensuremath{l}}

\newcommand{\freevariables}{\mathrm{freevar}}

\newcommand{\nfreevariables}{\ensuremath{t}}
\newcommand{\structurestring}{\ensuremath{\mathbb{S}}}

\DeclarePairedDelimiter{\abs}{\lvert}{\rvert}
\DeclarePairedDelimiter{\tuple}{(}{)}
\DeclarePairedDelimiter{\set}{\lbrace}{\rbrace}
\DeclarePairedDelimiter{\sequence}{\langle}{\rangle}
\DeclarePairedDelimiter{\bbset}{[}{]}
\DeclarePairedDelimiter{\dbset}{\llbracket}{\rrbracket}

\newcommand{\bset}[1]{\ensuremath{\bbset{#1}}}

\newcommand{\setst}{\xspace\colon\xspace}
\newcommand{\BigOh}{\ensuremath{\mathcal{O}}}
\newcommand{\astring}{\ensuremath{s}}
\newcommand{\bstring}{\ensuremath{\astring^{\prime}}}

\newcommand{\cartesianproduct}{\times}
\newcommand{\tensorproduct}{\otimes}

\newcommand{\fwidth}{\omega}
\newcommand{\dwidth}[2]{\ensuremath{\fwidth(#1,#2)}}

\newcommand{\state}{\ensuremath{q}}
\newcommand{\astate}{\state}
\newcommand{\bstate}{\ensuremath{\state^{\prime}}}

\newcommand{\layerleftstate}{\ensuremath{\mathfrak{p}}}
\newcommand{\layerrightstate}{\ensuremath{\mathfrak{q}}}
\newcommand{\avariable}{x}
\newcommand{\bvariable}{y}
\newcommand{\variable}{\avariable}

\newcommand{\formula}{\ensuremath{\psi}}
\newcommand{\MSOformula}{\ensuremath{\varphi}}

\newcommand{\freevariableset}{\ensuremath{X}}

\newcommand{\aset}{S}
\newcommand{\vertex}{\ensuremath{u}}
\newcommand{\avertex}{\vertex}
\newcommand{\bvertex}{\ensuremath{v}}
\newcommand{\anumber}{\ensuremath{c}}

\newcommand{\oddlength}{\ensuremath{k}}
\newcommand{\stringlength}{\ensuremath{k}}

\newcommand{\asymbol}{\ensuremath{\sigma}}
\newcommand{\bsymbol}{\ensuremath{\nu}}

\newcommand{\padasymbol}{\ensuremath{\widetilde{\asymbol}}}
\newcommand{\padbsymbol}{\ensuremath{\widetilde{\bsymbol}}}

\newcommand{\hypercubegraph}[1]{\ensuremath{\mathcal{H}_{#1}}}
\newcommand{\hypercubegraphfamily}{\ensuremath{\mathscr{H}}}
\newcommand{\hypercubefamily}{\ensuremath{\mathscr{H}}}

\newcommand{\ie}{i.e.\xspace}

\newcommand{\true}{1}
\newcommand{\false}{0}

\newcommand{\defeq}{\doteq}

\newcommand{\layertransition}{\ensuremath{\delta}}

\newcommand{\odddefiningset}[3]{\ensuremath{\layeralphabet{#1}{#2}^{\protect\circ{#3}}}}

\newcommand{\playerstructuredefiningset}[3]{\ensuremath{\layervocabularyalphabet{#1}{#2}{#3}^{\circledast}}}

\newcommand{\podddefiningset}[2]{\ensuremath{\layeralphabet{#1}{#2}^{\circledast}}}

\newcommand{\associatedrelation}[1]{\ensuremath{\mathsf{rel}(#1)}}
\newcommand{\associatedlang}[1]{\ensuremath{\mathsf{lang}(#1)}}
\newcommand{\isomorphism}{\pi}
\newcommand{\isomorphic}{\ensuremath{\simeq}}

\newcommand{\dcup}{\ensuremath{\uplus}}

\newcommand{\anote}[1]{\todo{#1}}
\newcommand{\rchanged}[1]{{\color{black}#1}} 

\newcommand{\classstructures}{\mathcal{C}}
\newcommand{\structuraltuple}{\ensuremath{\mathcal{D}}}
\newcommand{\structure}{\ensuremath{\mathfrak{A}}}
\newcommand{\derivedstructure}{\ensuremath{\mathfrak{s}}}
\newcommand{\astructure}{\structure}
\newcommand{\bstructure}{\ensuremath{\structure^{\prime}}}

\newcommand{\derivedrelation}{\mathcal{R}}

\newcommand{\alogic}{\mathcal{L}}
\newcommand{\atheory}{\mathcal{T}}

\newcommand{\svocabulary}[1][\alphabet]{\ensuremath{\varrho(#1)}}

\providecommand{\keywords}[1]{\noindent\textbf{\textrm{Keywords.}} #1}
\usepackage{fullpage}

\begin{document}

\title{On the Width of Regular Classes of Finite Structures\footnote{A preliminary version of this work was published in the proceedings of the 27th International Conference on Automated Deduction \cite{MeloOliveira2019}} 
}

\author{Alexsander Andrade de Melo$^1$ \hspace{0.6cm} Mateus de Oliveira Oliveira$^2$ \\ 
\\
$^1$Federal University of Rio de Janeiro, Rio de Janeiro, Brazil \\ 
\texttt{aamelo@cos.ufrj.br} \\ 
$^2$University of Bergen, Bergen, Norway \\
\texttt{mateus.oliveira@uib.no}
}

\date{}

\maketitle

\begin{abstract}
In this work, we introduce the notion of decisional width of a finite
relational structure and the notion of decisional width of a regular class of
finite structures.  Our main result states that given a first-order formula
$\formula$ over a vocabulary $\vocabulary$, and a finite automaton
$\finiteautomaton$ over a suitable alphabet $\layervocabularyalphabet{\alphabet}{\width}{\vocabulary}$
representing a width-$\width$ regular-decisional class of $\vocabulary$-structures
$\classstructures$, one can decide in time $f(\vocabulary,\alphabet,\formula,\width)\cdot |\finiteautomaton|$
whether some $\vocabulary$-structure in $\classstructures$ satisfies $\formula$. Here, $f$ is a function that 
depends on the parameters $\vocabulary,\alphabet,\formula,\width$, but not on the size of the automaton
$\finiteautomaton$ representing the class. Therefore, besides implying that the first-order theory of 
any given regular-decisional class of finite structures is decidable, it also implies that when 
the parameters $\vocabulary$, $\formula$, $\alphabet$ and $\width$ are fixed,
decidability can be achieved in linear time on the size of the input automaton $\finiteautomaton$. 	
Building on the proof of our main result, we show that 
the problem of counting satisfying assignments for a first-order logic formula in a given 
structure $\astructure$ of width $\width$ is fixed-parameter tractable with 
respect to $\width$, and can be solved in quadratic time on the length of 
the input representation of $\astructure$. \\ 

\keywords{Automatic Structures, Width Measures, First Order Logic}
\end{abstract}

\section{Introduction}
Relational structures can be used to formalize a wide variety of mathematical constructions, 
such as graphs, hypergraphs, groups, rings, databases, etc. 
Not surprisingly, relational structures are a central object of study in several subfields 
of computer science, such as database theory, learning theory \cite{getoor2001learning}, 
constraint satisfaction theory \cite{Bulatov2016graphs,KolaitisVardi2000}, 
and automated theorem proving \cite{poon2008general,sutskever2009modelling,Courcelle1990MSO,CourcelleMakowskyRotics2000,Zaid2017}. 

Given a class $\classstructures$ of finite relational structures and a logic $\alogic$,
such as first-order logic (FO) or monadic second-order logic (MSO), 
the $\alogic$ theory of $\classstructures$ 
is the set $\atheory(\alogic,\classstructures)$ of all logical sentences from $\alogic$ that
are satisfied by at least one structure in $\classstructures$. The theory $\atheory(\alogic,\classstructures)$
is said to be {\em decidable} if the problem of determining whether a given sentence $\varphi$ belongs to
$\atheory(\alogic,\classstructures)$ is decidable. Showing that the $\alogic$ theory of a particular 
class of structures $\classstructures$ is  decidable is an endeavour of fundamental importance because 
many interesting mathematical statements can be formulated as the problem of determining whether some/every/no structure in
a given class $\classstructures$ of finite structures satisfies a given logical sentence $\varphi$.

The problem of deciding logical properties of relational structures has been studied extensively 
in the field of automatic structure theory. Intuitively, a relational structure is automatic if its domain 
and each of its relations can be defined by finite automata. Although the idea of representing certain 
types of relational structures using automata dates back to seminal work of B\"{u}chi 
and Rabin, 
and it has been extensively applied in the field of automatic group theory \cite{farb1992automatic},
Khoussainov and Nerode were the first to define and investigate a general 
notion of automatic structure \cite{khoussainov1995automatic}. One early use 
of the notion of automatic structure is in providing a simplified proof 
of a celebrated theorem of Presburger stating that the first-order theory of $(\N,+)$
is decidable.
In \cite{Kruckman2012} Kruckman et al. introduced the notion of structure with advice and used it to show that 
the first-order theory of $(\Q,+)$ is decidable. The theory of automatic structures with advice was extended 
by Zaid, Gr\"adel and Reinhardt and used to define classes of relational structures. 
They showed that an automatic class of structures with advice has decidable first-order logic theory if and only
if the set of used advices has a decidable monadic second-order logic theory \cite{Zaid2017}.

Interesting classes of structures with decidable MSO theory can be obtained by combining automata theoretic
techniques with techniques from parameterized complexity theory. In particular, for each fixed $k\in \N$, if $\classstructures$ is the class of graphs of treewidth at 
most $k$, and $\alogic$ is the monadic second order logic of graphs with vertex-set and edge-set quantifications (MSO$_2$ logic),
then the $\atheory(\alogic,\classstructures)$ is decidable \cite{Courcelle1990MSO,Zaid2017}. 
Given the expressive power of MSO$_2$ logic, this decidability result implies 
that for each fixed $k\in \N$, the validity of several long-standing conjectures, 
such as Hadwiger's conjecture for $K_r$-free graphs (for each fixed $r\geq 7$) can be decided automatically when restricted to the class of graphs of 
treewidth at most $k$. Similar decidability results can be obtained when $\classstructures$ is the class 
of graphs of cliquewidth at most $k$, for fixed $k\in \N$, and $\alogic$ is the monadic second-order 
logic of graphs with vertex-set quantifications (MSO$_1$ logic) \cite{CourcelleMakowskyRotics2000,Zaid2017}.
On the other hand, it is known that the MSO$_1$ theory of any class of graphs that contain arbitrarily large grids
is undecidable. Additionally, a celebrated long-standing conjecture due to Seese states that the 
MSO$_1$ theory of any class of graphs of unbounded cliquewidth is undecidable \cite{Seese1991structure}.

\subsection{Our Contributions}

In this work, we introduce a new width measure for finite relational structures that builds on the
notion of ordered decision diagrams (ODDs). An ODD  over an alphabet $\alphabet$ is essentially an
acyclic finite automaton over $\alphabet$ where states are split into a sequence of frontiers, 
and transitions are split into a sequence of layers in such a way that transitions in layer $i$
send states in frontier $i-1$ to states in frontier $i$. We note that ODDs over a binary alphabet have
been extensively studied in the literature and are usually called  
{\em OBDDs (ordered binary decision diagrams)}. This formalism has been widely used in the field of symbolic computation 
due to its ability to concisely represent certain classes of combinatorial structures with exponentially many elements in the domain
\cite{Bollig2014width,Bollig2012symbolic,Hachtel1993symbolic,Woelfel2006symbolic,Sawitzki2004implicit}. 

An important complexity measure when dealing with ODDs is the notion of {\em width}. This is simply the maximum number
of states in a frontier of the ODD. In this work, we will be concerned with relational structures that can be represented 
using ODDs whose width is bounded by a constant $\width$. More precisely, we say that a relational structure
$\structure = \tuple{\structuredomain(\structure), \relationsymbol_1(\structure), \ldots, \relationsymbol_{\nrelations}(\structure)}$
over a relational vocabulary $\vocabulary = \tuple{\relationsymbol_1, \ldots, \relationsymbol_{\nrelations}}$ is 
$(\alphabet,\width)$-decisional if its domain $\structuredomain(\structure)$ can be represented by an ODD of width at most $\width$
over the alphabet $\alphabet$, and each of its relations $\relationsymbol_i(\structure)\subseteq \structuredomain(\structure)^{\arity_i}$
can be represented by an ODD of width at most $\width$ over the alphabet $\alphabet^{\otimes \arity_i}$, i.e. the $\arity_i$-fold
Cartesian product of $\alphabet$. If we let $\layeralphabet{\alphabet}{\width}$ denote the set of all layers whose frontiers have 
size at most $\width$, then an ODD over $\alphabet$ can be represented as a word over the alphabet $\layeralphabet{\alphabet}{\width}$. 
Similarly, an ODD over $\alphabet^{\otimes \arity_i}$ of width at most $\width$ can be represented as a word over the alphabet $\layeralphabet{\alphabet^{\otimes \arity_i}}{\width}$. 
Therefore, a $(\alphabet,\width)$-relational structure can be represented as a word over the alphabet 
$\layervocabularyalphabet{\alphabet}{\width}{\vocabulary} \defeq \layeralphabet{\alphabet}{\width} \tensorproduct \layeralphabet{\alphabet^{\tensorproduct\arity_1}}{\width} \tensorproduct \cdots \tensorproduct \layeralphabet{\alphabet^{\tensorproduct \arity_{\nrelations}}}{\width}$. Following this point of view, one can define classes of finite $(\alphabet,\width)$-structures over a vocabulary $\vocabulary$
using languages over the alphabet $\layervocabularyalphabet{\alphabet}{\width}{\vocabulary}$. In particular, we will be concerned with classes of structures that 
can be defined using regular languages over $\layervocabularyalphabet{\alphabet}{\width}{\vocabulary}$. We say that that a class $\classstructures$ of finite relational structures
is $(\alphabet,\width)$-regular-decisional
if there is a finite automaton $\finiteautomaton$ over $\layervocabularyalphabet{\alphabet}{\width}{\vocabulary}$ whose words represent precisely those structures in $\classstructures$. 

In this work, we show that given a first order sentence $\formula$ over a vocabulary $\vocabulary$, an alphabet $\alphabet$ and a constant $\width$, 
the class of $(\alphabet,\width)$-decisional structures satisfying $\formula$ is $(\alphabet,\width)$-regular-decisional. Building on this result
we show that given a finite automaton $\finiteautomaton$ over the alphabet $\layervocabularyalphabet{\alphabet}{\width}{\vocabulary}$ representing 
a class $\classstructures$ of structures, one can decide in time  $f(\vocabulary,\alphabet,\formula,\width)\cdot |\finiteautomaton|$ whether 
some structure defined by a word in $\lang(\finiteautomaton)$ satisfies $\varphi$. Finally, we show that given a formula $\formula$ over $\vocabulary$ with $\nfreevariables$ variables
$\avariable_1,\dots,\avariable_{\nfreevariables}$, and a presentation of a $\vocabulary$-structure $\astructure$ as a tuple of ODDs of length $\oddlength$ and width $\width$, one can count in time 
$f(\vocabulary,\alphabet,\formula,\nfreevariables,\width)\cdot k$ the number of assignments for $\avariable_1,\dots,\avariable_{\nfreevariables}$ that satisfy $\formula$ on $\astructure$.
As a straightforward application of this result, we have that several counting problems 
that are $\#W[1]$-complete or $\#W[2]$-complete on general graphs, 
such as counting the number of independent sets of size $r$ or the number of
dominating sets of size $r$ (for fixed $r$), can be solved in FPT time on $(\alphabet,\width)$-decisional graphs, provided a presentation is given at the input.

An interesting feature of our width measure for classes of structures is that it behaves very differently from 
usual complexity measures for graphs studied in structural complexity theory. As an example of this fact, we 
note that the family of hypercube graphs has regular-decisional width $2$, while it has unbounded width in 
many of the studied measures, such as treewidth and cliquewidth. This class of graphs has also unbounded degeneracy, 
and therefore it is not no-where dense. As a consequence, we have that most algorithmic meta-theorems proved so
far dealing with the interplay between first-order logic and structural graph theory 
\cite{grohe2008algorithmic,kreutzer2008algorithmic,grohe2014algorithmic} fail on graphs of constant
decisional-width. In this regard, it is interesting to note that our first-order decidability result cannot be 
generalized to MSO$_1$ logic, since the class of grid graphs has constant decisional width, while the MSO$_1$ theory
of grids is known to be undecidable~\cite{Seese1991structure,Hlinveny2006}.

\section{Preliminaries}
We denote by $\N \defeq \set{0, 1, \ldots}$ the set of natural numbers (including zero), 
and by $\pN \defeq \N\setminus \set{0}$ the set of positive natural numbers.
For each $\anumber \in \pN$, we let 
$\bset{\anumber} \defeq \set{1, 2, \ldots, \anumber}$  and  
$\dbset{\anumber} \defeq \set{0, 1, \ldots, \anumber-1}$. 
Given a set $S$ and a number $a\in \pN$, we let $S^{\cartesianproduct a}$ be the 
set of all $a$-tuples of elements from $S$.

\paragraph{\bf Relational Structures.}
\label{subsection:RelationalStructures}

A \emph{relational vocabulary} is a tuple  $\vocabulary = \tuple{\relationsymbol_1, \ldots, \relationsymbol_{\nrelations}}$ 
of \emph{relation symbols} where for each $i\in [\nrelations]$, the relation symbol $\relationsymbol_{i}$ 
is associated with an arity $\arityrelation_i\in\pN$.
Let $\vocabulary = \tuple{\relationsymbol_1, \ldots, \relationsymbol_{\nrelations}}$ be a relational vocabulary. 
A finite \emph{$\vocabulary$-structure} is a tuple $\astructure = \tuple{\structuredomain(\astructure), \relation_1(\astructure), \ldots, \relation_\nrelations(\astructure)}$ such that
\begin{enumerate}
  \item $\structuredomain(\astructure)$ is a non-empty finite set, called the \emph{domain} of $\astructure$;
  \item for each $i\in\bset{\nrelations}$, $\relation_i(\astructure) \subseteq \structuredomain(\astructure)^{\cartesianproduct \arityrelation_i}$ is an $\arityrelation_i$-ary relation.
\end{enumerate}

Let $\astructure$ and $\bstructure$ be $\vocabulary$-structures.
An \emph{isomorphism from $\astructure$ to $\bstructure$} is a bijection $\isomorphism \colon \structuredomain(\astructure)\rightarrow \structuredomain(\bstructure)$ such that, for each  $i \in \bset{\nrelations}$, $\tuple{\vertex_{1},\ldots,\vertex_{\arity_{i}}} \in \relationsymbol_{i}(\astructure)$ if and only if $\tuple{\isomorphism(\vertex_{1}),\dots,\isomorphism(\vertex_{\arity_{i}})} \in \relationsymbol_{i}(\bstructure)$. 
If there exists an isomorphism from $\astructure$ to $\bstructure$, then we say that $\astructure$ is \emph{isomorphic} to $\bstructure$, and we denote this fact by $\astructure \isomorphic \bstructure$.

\paragraph{\bf First-Order Logic.} Now, we briefly recall some basic concepts from first-order logic.
Extensive treatments of this subject can be found in \cite{Ebbinghaus2005,EbbinghausFumThomas2013}.

Let $\vocabulary = \tuple{\relationsymbol_1, \ldots, \relationsymbol_{\nrelations}}$ be a relational vocabulary.
We denote by $\fol{\vocabulary}$ the set of all \emph{first-order logic formulas over $\vocabulary$}, \ie the logic formulas comprising: variables to be used as placeholders for elements from the domain of a $\vocabulary$-structure; the Boolean connectives $\vee$, $\wedge$, $\neg$, $\rightarrow$ and $\leftrightarrow$; the quantifiers $\exists$ and $\forall$ that can be applied to the variables;
and the atomic logic formulas $\avariable=\bvariable$, where $\avariable$ and $\bvariable$ are variables, and $\arelation_i(\avariable_1,\ldots,\avariable_{\arityrelation_i})$ for some $i\in \bset{\nrelations}$, where $\avariable_1, \ldots, \avariable_{\arityrelation_i}$ are variables.

A variable $\avariable$ is said to be {\em free} in a formula $\fol{\vocabulary}$ if some of its occurrences 
is outside the scope of any quantifier. We let $\freevariables(\formula)$ denote the set of variables that 
are free in $\formula$. We write $\formula(\avariable_1,\ldots, \avariable_{\nfreevariables})$ to indicate
that $\freevariables(\formula) \subseteq \set{\avariable_1,\ldots, \avariable_{\nfreevariables}}$.
A \emph{sentence} is a formula without free variables. 

Let $\structure = \tuple{\structuredomain(\structure), \relationsymbol_1(\structure), \ldots, \relationsymbol_{\nrelations}(\structure)}$ be a finite $\vocabulary$-structure, $\formula(\variable_1, \ldots, \variable_\nfreevariables) \in \fol{\vocabulary}$ and $\vertex_1, \ldots, \vertex_{\nfreevariables}\in\structuredomain(\structure)$.
We write $\structure \models \formula[\avertex_1, \ldots, \avertex_\nfreevariables]$ to mean that $\structure$ satisfies $\formula$ when all the free occurrences of the variables $\avariable_1, \ldots, \avariable_\nfreevariables$ are interpreted by the values $\avertex_1, \ldots, \avertex_\nfreevariables$, respectively.
In particular, if $\formula$ is a sentence, then we may write $\structure \models \formula$ to mean that $\structure$ satisfies $\formula$.
In this case, we also say that $\structure$ is a \emph{model} of $\formula$.
If $\formula(\avariable_1, \ldots, \avariable_\nfreevariables) \equiv \relationsymbol_{i}(\avariable_{\cmap(1)}, \ldots, \avariable_{\cmap(\arityrelation_{i})})$ for some mapping $\cmap \colon \bset{\arityrelation_{i}} \rightarrow \bset{\nfreevariables}$ and some $i \in \bset{\nrelations}$, then $\structure \models \formula[\vertex_{1}, \ldots, \vertex_{\nfreevariables}]$ if and only if $\tuple{\avertex_{\cmap(1)}, \ldots, \avertex_{\cmap(\arityrelation_{i})}}\in\relationsymbol_{i}(\structure)$.
The semantics of the equality symbol, of the quantifiers $\exists$ and $\forall$, and of the Boolean connectives $\vee$, $\wedge$, $\neg$, $\rightarrow$ and $\leftrightarrow$ are the usual ones.

\paragraph{\bf Languages.}
An \emph{alphabet} is any finite, non-empty set of symbols $\alphabet$.
A \emph{string} over $\alphabet$ is any finite sequence of symbols from $\alphabet$.
We denote by $\alphabet^*$ the set of all strings over $\alphabet$, including 
the empty string $\varepsilon$, and by $\alphabet^{+}$ the set of all (non-empty) strings over $\alphabet$.
A \emph{language} over $\alphabet$ is any subset $\lang$ of $\alphabet^{*}$.
In particular, for each $\stringlength \in \pN$, we let $\alphabet^{\stringlength}$ be the language of all strings of length $\stringlength$ over $\alphabet$, and we let $\alphabet^{\leq \stringlength} \defeq \alphabet^{1} \cup \cdots \cup \alphabet^{\stringlength}$ be the language of all (non-empty) strings of length at most $\stringlength$ over $\alphabet$. 

For each alphabet $\alphabet$, we let $\paddingsymbol$ be a special symbol, called the {\em padding symbol}, such that $\paddingsymbol \not\in \alphabet$, and we write $\alphabet \dcup \set{\paddingsymbol}$ to denote the disjoint union between $\alphabet$ and $\set{\paddingsymbol}$.
For each $\asymbol \in \alphabet \dcup \set{\paddingsymbol}$ and each $\anumber \in \pN$, we let $\asymbol^{\cartesianproduct \anumber}$ denote the tuple $\tuple{\asymbol, \ldots, \asymbol}$ composed by $\anumber$ copies of the symbol $\asymbol$.

\paragraph{\bf Tensor Product.}
\label{subsection:TensorProduct}

Let $\alphabet_1, \ldots, \alphabet_{\anumber}$ be $\anumber$ alphabets, where $\anumber \in \pN$.
The \emph{tensor product} of $\alphabet_1, \ldots, \alphabet_{\anumber}$ is defined as the alphabet $$\alphabet_1 \tensorproduct \cdots \tensorproduct \alphabet_{\anumber} \defeq \set{\tuple{\asymbol_1, \ldots, \asymbol_{\anumber}} \setst \asymbol_i\in\alphabet_i,\; i\in\bset{\anumber}}\text{.}$$
In particular, for each $\anumber\in\pN$ and each alphabet $\alphabet$, we define the \emph{$\anumber$-th tensor power} of $\alphabet$ as the alphabet $\alphabet^{\tensorproduct \anumber} \defeq 
\overbrace{\alphabet \tensorproduct \cdots \tensorproduct \alphabet}^{\anumber\;\mathit{times}}\text{.}$
Note that the tensor product of $\alphabet_1, \ldots, \alphabet_{\anumber}$ (the $\anumber$-th tensor power of $\alphabet$) 
may be simply regarded as the Cartesian product of 
$\alphabet_1, \ldots, \alphabet_{\anumber}$ (the $\anumber$-ary Cartesian power of $\alphabet$, respectively).
Also, note that $\alphabet^{\tensorproduct 1} = \alphabet$.

For each $i\in\bset{\anumber}$, let $\astring_i = \asymbol_{i,1} \cdots \asymbol_{i,\stringlength_{i}}$ be a string of length $\stringlength_{i}$, where $\stringlength_{i} \in \pN$, over the alphabet $\alphabet_{i}$.
The \emph{tensor product} of $\astring_1, \ldots, \astring_{\anumber}$ is defined as the string $$\astring_{1} \tensorproduct \cdots \tensorproduct \astring_{\anumber} \defeq \tuple{\padasymbol_{1,1}, \ldots, \padasymbol_{\anumber,1}} \cdots \tuple{\padasymbol_{1,\stringlength}, \ldots, \padasymbol_{\anumber,\stringlength}}\text{}$$
of length $\stringlength = \max\set{\stringlength_{1},\ldots,\stringlength_{\anumber}}$ over the alphabet $(\alphabet_1 \dcup \set{\paddingsymbol}) \tensorproduct \cdots \tensorproduct (\alphabet_{\anumber} \dcup \set{\paddingsymbol})$ such that, for each $i \in \bset{\anumber}$ and each $j \in \bset{\stringlength}$, 
\begin{equation*}
  \padasymbol_{i,j} \defeq \left\{
  \begin{array}{lr}
    \asymbol_{i,j} & \text{ if } j \leq \stringlength_{i} \\
    \paddingsymbol & \text{ otherwise.}
  \end{array}
  \right.
\end{equation*}
For instance, the tensor product of the strings $aabab$ and $abb$ over the alphabet $\set{a,b}$ is the string $\tuple{a,a}\tuple{a,b}\tuple{b,b}\tuple{a,\paddingsymbol}\tuple{b,\paddingsymbol}$ over the alphabet $(\set{a,b}\dcup\set{\paddingsymbol})^{\tensorproduct 2}$. 

For each $i\in\bset{\anumber}$, let $\lang_i \subseteq \alphabet_i^{+}$ be a language over the alphabet $\alphabet_{i}$. The \emph{tensor product} of $\lang_1, \ldots, \lang_{\anumber}$ is defined as the language
\begin{equation*}
  \lang_1 \tensorproduct \cdots \tensorproduct \lang_{\anumber} \defeq \set{\astring_1 \tensorproduct \cdots \tensorproduct \astring_{\anumber} \setst \astring_i\in\lang_i, i\in\bset{\anumber}}\text{.}
\end{equation*}

We remark that in the literature~\cite{Blumensath1999,Blumensath2000,KhoussainovMinnes2007three,Zaid2017}
the tensor product of strings (as well as the tensor product of languages) is also commonly called of \emph{convolution}.

\paragraph{\bf Finite Automata.}
\label{subsection:FiniteAutomata}

A \emph{finite automaton} is a tuple $\finiteautomaton = \tuple{\alphabet,\automatonstates,\automatontransitions,\automatoninitialstates,\automatonfinalstates}$, where $\alphabet$ is an alphabet, $\automatonstates$ is a finite set of \emph{states}, ${\automatontransitions\subseteq \automatonstates\times \alphabet\times \automatonstates}$  is a set of \emph{transitions}, $\automatoninitialstates\subseteq \automatonstates$ is a set of \emph{initial states} and $\automatonfinalstates\subseteq \automatonstates$ is a set of \emph{final states}.
The size of a finite automaton $\finiteautomaton$ is defined as 
$\abs{\finiteautomaton} \defeq \abs{\automatonstates} + \abs{\automatontransitions}\log |\alphabet|$.

Let $\stringlength\in\pN$ and $\astring = \asymbol_{1}\cdots\asymbol_{\stringlength}\in\alphabet^{\stringlength}$.  
We say that $\finiteautomaton$ \emph{accepts} $\astring$ if there exists a sequence of transitions $\sequence{\tuple{\state_{0},\asymbol_{1},\state_{1}},\tuple{\state_{1},\asymbol_{2},\state_{2}}, \ldots,\tuple{\state_{\stringlength-1},\asymbol_\stringlength,\state_\stringlength}}$, called an \emph{accepting sequence} for $\astring$ in $\finiteautomaton$, such that $\state_{0}\in \automatoninitialstates$, $\state_\stringlength\in\automatonfinalstates$ and, for each $i\in\bset{\stringlength}$, ${\tuple{\state_{i-1},\asymbol_{i},\state_{i}}\in\automatontransitions}$. 
The language of $\finiteautomaton$, denoted by $\automatonlang{\finiteautomaton}$, is defined as the set of all strings accepted by $\finiteautomaton$, \ie $\automatonlang{\finiteautomaton} \defeq \set*{\astring \in \alphabet^{+} \setst \astring \text{ is accepted by } \finiteautomaton}\text{.}$

\paragraph{\bf Regular Relations.}\label{subsection:AutomaticStructuresStructures}

Let $\alphabet$ be an alphabet.
A language $\lang \subseteq \alphabet^{+}$ is called \emph{regular} if there exists a finite automaton $\finiteautomaton$ over $\alphabet$ such that $\automatonlang{\finiteautomaton} = \lang$.

For each $\arityrelation \in \pN$ and each language $\lang \subseteq (\alphabet^{\tensorproduct \arityrelation})^{+}$ over the alphabet $\alphabet^{\tensorproduct \arityrelation}$, we let $\associatedrelation{\lang} \defeq \set{\tuple{\astring_{1},\ldots,\astring_{\arityrelation}} \setst \astring_{1}\tensorproduct \cdots \tensorproduct \astring_{\arityrelation} \in \lang}$ be the \emph{relation associated with $\lang$}.
On the other hand, for each $\arityrelation \in \pN$ and each $\arityrelation$-ary relation $\relation \subseteq (\alphabet^{+})^{\cartesianproduct \arityrelation}$, we let $\associatedlang{\relation} \defeq \set{\astring_{1}\tensorproduct \cdots \tensorproduct \astring_{\arityrelation} \setst \tuple{\astring_{1},\ldots,\astring_{\arityrelation}} \in \relation}$ be the \emph{language associated with $\relation$}.
We say that such a relation $\relation$ is \emph{regular} if its associated language $\associatedlang{\relation}$ is a regular
subset of $(\alphabet^{\tensorproduct \arityrelation})^+$.

\section{Ordered Decision Diagrams}
\label{section:ODDs}

Let $\alphabet$ be an alphabet and $\width\in\pN$.  
A \emph{$\tuple{\alphabet,\width}$-layer} is a tuple $\alayer \defeq \tuple{\layerleftfrontier,\layerrightfrontier, \layertransitions, \layerinitialstates, \layerfinalstates, \layerinitialflag, \layerfinalflag}$, where $\layerleftfrontier \subseteq \dbset{\width}$ is a set of \emph{left states}, $\layerrightfrontier \subseteq \dbset{\width}$ is a set of \emph{right states}, $\layertransitions \subseteq \layerleftfrontier \cartesianproduct
(\alphabet \dcup \set{\paddingsymbol}) \cartesianproduct \layerrightfrontier$ is a set of \emph{transitions}, $\layerinitialstates \subseteq \layerleftfrontier$ is a set of \emph{initial states}, $\layerfinalstates \subseteq \layerrightfrontier$ is a set of \emph{final states} and $\layerinitialflag, \layerfinalflag \in \set{\false, \true}$ are Boolean flags satisfying the following conditions: 
\begin{enumerate*}[label=\roman*)]
  \item $\layerinitialstates =\emptyset$ if $\layerinitialflag = \false$, and 
  \item $\layerfinalstates = \emptyset$ if $\layerfinalflag = \false$.
\end{enumerate*} 
We remark that it is possible that $\layerinitialflag = \true$ and $\layerinitialstates =\emptyset$. Similarly, it might be the case in which $\layerfinalflag = \true$ and $\layerfinalstates = \emptyset$.

In what follows, we may write $\layerleftfrontier(\alayer)$, $\layerrightfrontier(\alayer)$, $\layertransitions(\alayer)$, $\layerinitialstates(\alayer)$, $\layerfinalstates(\alayer)$, $\layerinitialflag(\alayer)$ and $\layerfinalflag(\alayer)$ to refer to the sets $\layerleftfrontier$, $\layerrightfrontier$, $\layertransitions$, $\layerinitialstates$ and $\layerfinalstates$ and to the Boolean flags $\layerinitialflag$ and $\layerfinalflag$, respectively.

We let $\layeralphabet{\alphabet}{\width}$ denote the set of all $\tuple{\alphabet,\width}$-layers. 
Note that, $\layeralphabet{\alphabet}{\width}$ is non-empty and has at most $2^{\BigOh(\abs{\alphabet} \cdot \width^2)}$ elements. 
Therefore, $\layeralphabet{\alphabet}{\width}$ may be regarded as an alphabet.

Let $\oddlength\in \pN$. A \emph{$(\alphabet,\width)$-ordered decision diagram} (or simply, $(\alphabet,\width)$-\emph{ODD}) of \emph{length} $\oddlength$ is a string $\aodd \defeq \alayer_1 \cdots \alayer_\oddlength \in \layeralphabet{\alphabet}{\width}^{\oddlength}$ of length $\oddlength$ over the alphabet $\layeralphabet{\alphabet}{\width}$ satisfying the following conditions:
\begin{enumerate}
  \item for each $i\in\bset{\oddlength-1}$, $\layerleftfrontier(\alayer_{i+1}) = \layerrightfrontier(\alayer_i)$;
  \item $\layerinitialflag(\alayer_{1}) = \true$ and, for each $i \in \set{2, \ldots, \oddlength}$, $\layerinitialflag(\alayer_{i}) = \false$;
  \item $\layerfinalflag(\alayer_{\oddlength}) = \true$ and, for each $i \in \bset{\oddlength-1}$, $\layerfinalflag(\alayer_{i}) = \false$.
\end{enumerate}
The \emph{width} of $\aodd$ is defined as $\fwidth(\aodd) \defeq \max\set*{\abs{\layerleftfrontier(\alayer_{1})}, \ldots, \abs{\layerleftfrontier(\alayer_{\oddlength})}, \abs{\layerrightfrontier(\alayer_{\oddlength})}}\text{.}$
We remark that $\fwidth(\aodd) \leq \width$.

For each $\oddlength \in \pN$, we denote by $\odddefiningset{\alphabet}{\width}{\oddlength}$ the set of all $(\alphabet,\width)$-ODDs of length $\oddlength$. 
And, more generally, we let $\podddefiningset{\alphabet}{\width}$ be the set of all $(\alphabet,\width)$-ODDs, \ie
$\podddefiningset{\alphabet}{\width} \defeq \bigcup_{k\in \pN}  \odddefiningset{\alphabet}{\width}{\oddlength}$.

Let $\odd = \layer_1 \cdots \layer_{\oddlength} \in \odddefiningset{\alphabet}{\width}{\oddlength}$ and $\astring = \asymbol_{1} \cdots \asymbol_{\oddlength^{\prime}} \in \alphabet^{\leq \oddlength}$. 
A \emph{valid sequence} for $\astring$ in $\odd$ is a sequence of transitions $\sequence{\tuple{\layerleftstate_1,\padasymbol_{1},\layerrightstate_1}, \ldots, \tuple{\layerleftstate_{\oddlength},\padasymbol_{\oddlength},\layerrightstate_{\oddlength}}}$ satisfying the following conditions:
\begin{enumerate}
  \item for each $i\in\bset{\oddlength}$, $\padasymbol_{i} = \asymbol_{i}$ if $i \leq \oddlength^{\prime}$, and $\padasymbol_{i} = \paddingsymbol$ otherwise;
  \item for each $i\in\bset{\oddlength}$, $\tuple{\layerleftstate_{i}, \padasymbol_{i}, \layerrightstate_{i}}\in\layertransitions(\alayer_{i})$;
  \item for each $i\in\bset{\oddlength-1}$, $\layerleftstate_{i+1} = \layerrightstate_{i}$.
\end{enumerate}
Such a sequence is called \emph{accepting} for $\astring$ if, in addition to the above conditions, $\layerleftstate_{1} \in \layerinitialstates(\alayer_1)$ and $\layerrightstate_{\oddlength} \in \layerfinalstates(\alayer_{\oddlength})$.
We say that $\odd$ \emph{accepts} $\astring$ if there exists an accepting sequence for $\astring$ in $\aodd$. 
The language of $\aodd$, denoted by $\oddlang{\aodd}$, is defined as the set of all strings accepted by $\aodd$, \ie $\oddlang{\aodd} \defeq \set*{\astring \in \alphabet^{\leq \oddlength} \setst \astring \text{ is accepted by } \aodd}\text{.}$

\begin{proposition}\label{proposition:odd_length}
Let $\alphabet$ be an alphabet, $\width, \oddlength \in \pN$ and $\aodd \in \odddefiningset{\alphabet}{\width}{\oddlength}$.
For each natural number $\oddlength^{\prime} \geq \oddlength$, there exists an ODD $\bodd \in \odddefiningset{\alphabet}{\width}{\oddlength^{\prime}}$\!\! such that $\oddlang{\bodd}=\oddlang{\aodd}$.
\end{proposition}
\begin{proof}
  Assume that $\aodd = \alayer_{1}\cdots\alayer_{\oddlength} \in \odddefiningset{\alphabet}{\width}{\oddlength}$.
    If $\oddlength^{\prime} = \oddlength$, then simply define $\bodd = \aodd$.
    Thus, assume that $\oddlength^{\prime} > \oddlength$, and let $\bodd = \blayer_{1}\cdots\blayer_{\oddlength}\blayer_{\oddlength+1}\cdots\blayer_{\oddlength^{\prime}}$ be the ODD in $\odddefiningset{\alphabet}{\width}{\oddlength^{\prime}}$ defined as follows.
    \begin{itemize}
      \item For each $i \in \bset{\oddlength - 1}$, we let $\blayer_{i} = \alayer_{i}$.
      \item We let $\blayer_{\oddlength}$ be the layer in $\layeralphabet{\alphabet}{\width}$ obtained from $\alayer_{\oddlength}$ by setting $\layerfinalstates(\blayer_{\oddlength}) = \emptyset$ and $\layerfinalflag(\blayer_{\oddlength}) = \false$.
      \item We let $\blayer_{\oddlength+1}$ be the layer in $\layeralphabet{\alphabet}{\width}$ defined as follows: $\layerleftfrontier(\blayer_{\oddlength+1}) = \layerfinalstates(\alayer_{\oddlength})$; $\layerrightfrontier(\blayer_{\oddlength+1}) = \set{0}$; $\layertransitions(\blayer_{\oddlength+1}) = \set{\tuple{\layerleftstate, \paddingsymbol, 0} \setst \layerleftstate \in \layerfinalstates(\alayer_{\oddlength})}$; $\layerinitialstates(\blayer_{\oddlength+1}) = \emptyset$; $\layerfinalstates(\blayer_{\oddlength+1}) = \emptyset$ if $\oddlength^{\prime} > \oddlength+1$, and $\layerfinalstates(\blayer_{\oddlength+1}) = \set{0}$ otherwise; $\layerinitialflag(\blayer_{\oddlength+1}) = \false$; and $\layerfinalflag(\blayer_{\oddlength+1}) = \false$ if $\oddlength^{\prime} > \oddlength+1$, and $\layerfinalflag(\blayer_{\oddlength+1}) = \true$ otherwise.
      \item Finally, for each $i\in \set{\oddlength+2, \ldots, \oddlength^{\prime}}$, we let $\blayer_{i}$ be the layer in $\layeralphabet{\alphabet}{\width}$ defined as follows: $\layerleftfrontier(\blayer_{i}) = \layerrightfrontier(\blayer_{i}) = \set{0}$; $\layertransitions(\blayer_{i}) = \set{\tuple{0, \paddingsymbol, 0}}$; $\layerinitialstates(\blayer_{i}) = \emptyset$; $\layerfinalstates(\blayer_{i}) = \emptyset$ if $i < \oddlength^{\prime}$, and $\layerfinalstates(\blayer_{i}) = \set{0}$ otherwise; $\layerinitialflag(\blayer_{i}) = \false$; and $\layerfinalflag(\blayer_{i}) = \false$ if $i < \oddlength^{\prime}$, and $\layerfinalflag(\blayer_{i}) = \true$ otherwise.
    \end{itemize}
One can verify that $\bodd$ is in fact an ODD in $\odddefiningset{\alphabet}{\width}{\oddlength^{\prime}}$ and $\oddlang{\bodd} = \oddlang{\aodd}$.
\end{proof}

\section{Regular-Decisional Classes of Finite Relational Structures}
\label{section:RegularDecisional}

In this section we combine ODDs with finite automata in order to define infinite classes 
of finite structures. We start by defining the notion of $(\alphabet,\width,\vocabulary)$-structural 
tuples of ODDs as a way to encode single finite $\vocabulary$-structures. Subsequently, we use finite 
automata over a suitable alphabet in order to uniformly define infinite families of finite $\vocabulary$-structures.

\subsection{Decisional Relations and Decisional Relational Structures}
\label{subsection:DecisionalStructures}

Let $\alphabet$ be an alphabet and $\arityrelation, \width \in \pN$. 
A \emph{finite} \mbox{$\arityrelation$-ary} relation $\relation \subseteq (\alphabet^{+})^{\cartesianproduct \arityrelation}$ 
is said 
to be strongly  $(\alphabet,\width)$-decisional if there exists an ODD $\odd$ in $\podddefiningset{\alphabet^{\tensorproduct \arityrelation}}{\width}$ such that $\associatedrelation{\oddlang{\odd}} = \relation$. Given a set $\aset$, we say that a relation 
$\relation \subseteq \aset^{\cartesianproduct \arityrelation}$ is $(\alphabet,\width)$-decisional if $\relation$ is isomorphic 
to some strongly $(\alphabet,\width)$-decisional structure.

Let $\vocabulary = \tuple{\relationsymbol_1,\dots,\relationsymbol_{\nrelations}}$ be a relational vocabulary, $\alphabet$ be an alphabet and $\width \in \pN$. 
We say that a tuple $\structuraltuple = \tuple{\odd_{0},\odd_{1},\dots,\odd_{\nrelations}}$ of ODDs is 
\emph{$(\alphabet,\width,\vocabulary)$-structural} if there exists a positive natural number $\oddlength\in \pN$, called the \emph{length of $\structuraltuple$}, such that the following conditions are satisfied:
\begin{enumerate}
  \item $\odd_{0}$ is an ODD in $\odddefiningset{\alphabet}{\width}{\oddlength}$;
  \item for each $i\in \bset{\nrelations}$, $\odd_{i}$ is an ODD in $\odddefiningset{\alphabet^{\tensorproduct \arity_i}}{\width}{\oddlength}$;
  \item for each $i\in \bset{\nrelations}$, $\oddlang{\odd_{i}}\subseteq \oddlang{\odd_{0}}^{\tensorproduct \arityrelation_{i}}$. 
\end{enumerate}

Let $\structuraltuple = \tuple{\odd_{0},\odd_{1},\ldots,\odd_{\nrelations}}$ be a $(\alphabet,\width,\vocabulary)$-structural tuple. 
The \emph{$\vocabulary$-structure derived from $\structuraltuple$} is defined as the finite $\vocabulary$-structure $$\derivedstructure(\structuraltuple) \defeq \tuple{\structuredomain(\derivedstructure(\structuraltuple)),\relationsymbol_{1}(\derivedstructure(\structuraltuple)),\ldots,\relationsymbol_{\nrelations}(\derivedstructure(\structuraltuple))}\text{,}$$ with domain $\structuredomain(\derivedstructure(\structuraltuple)) = \oddlang{\odd_0}$, such that $\relationsymbol_{i}(\derivedstructure(\structuraltuple)) = \associatedrelation{\oddlang{\odd_{i}}}$ for each $i\in \bset{\nrelations}$.

We say that a finite $\vocabulary$-structure $\structure$  is \emph{strongly $(\alphabet,\width)$-decisional} if there exists some
$(\alphabet,\width,\vocabulary)$-structural tuple $\structuraltuple$ such that 
$\structure  = \derivedstructure(\structuraltuple)$. We say that a finite $\vocabulary$-structure $\structure$,
whose domain is not necessarily a subset of $\alphabet^+$, is {\em $(\alphabet,\width)$-decisional} if $\structure$ is isomorphic to some strongly $(\alphabet,\width)$-decisional structure.

The \emph{$\alphabet$-decisional width} of a finite $\vocabulary$-structure $\structure$, denoted by $\dwidth{\alphabet}{\structure}$,
is defined as the minimum $\width \in \pN$ such that $\structure$ is $(\alphabet,\width)$-decisional. The following proposition
states that if a relation is $(\alphabet,\width)$-decisional, then it is also $(\set{0,1},\width^{\prime})$-decisional
for a suitable $\width^{\prime} \in \pN$.

\begin{proposition}\label{proposition:odd_alphabet}
  Let $\vocabulary = \tuple{\relationsymbol_{1}, \ldots, \relationsymbol_{\nrelations}}$ be a relational vocabulary, $\alphabet$ be an alphabet and $\width \in \pN$.
  If $\structure$ is a finite $\vocabulary$-structure with $\dwidth{\alphabet}{\structure} \leq \width$ for some $\width \in \pN$, then 
$\dwidth{\set{0,1}}{\structure} \leq \width^{2}\cdot\abs{\alphabet}^{\max\{1,\arityrelation_1,\ldots,\arityrelation_{\nrelations}\}}$.
\end{proposition}
\begin{proof}
  Since by hypothesis $\dwidth{\alphabet}{\structure} \leq \width$, there exists a strongly $(\alphabet,\width)$-decisional $\vocabulary$-structure that is isomorphic to $\structure$.
  Thus, for simplicity, assume without loss of generality that $\structure$ is strongly $(\alphabet,\width)$-decisional.
  For each $i \in \dbset{\nrelations+1}$, let $\odd_{i} \in \podddefiningset{\alphabet^{\tensorproduct \arityrelation_{i}}}{\width}$ such that $\associatedrelation{\oddlang{\odd_{i}}} = \relationsymbol_{i}(\structure)$, where $\arityrelation_{0} = 1$.
  Based on Proposition~\ref{proposition:odd_length}, we can assume without loss of generality that the ODDs $\odd_{0}, \odd_{1}, \ldots, \odd_{\nrelations}$ have the same length, say $\oddlength \in \pN$.
  Thus, assume that $\odd_{i} = \layer_{i,1}\cdots\layer_{i,\oddlength}$ for each $i \in \dbset{\nrelations+1}$.

  Let $\anumber = \lceil\log_{2}\abs{\alphabet}\rceil$ and $\amap \colon (\alphabet \dcup \set{\paddingsymbol}) \rightarrow (\set{0,1} \dcup \set{\paddingsymbol})^{\anumber}$ be an injection, 
\anote{(R1.11)}\rchanged{which sends each symbol in $\alphabet\dcup\set{\paddingsymbol}$ to a string of length $\anumber$ over $\set{0,1}\dcup\set{\paddingsymbol}$,} such that $\amap(\paddingsymbol) \in \set{\paddingsymbol}^{\anumber}$ and, for each $\asymbol \in \alphabet$, $\amap(\asymbol) \in \set{0,1}^{\anumber}$. 
  Note that, each symbol in $\alphabet$ can be associated with an exclusive number in $\dbset{\abs{\alphabet}}$.
  Thus, we might, for instance, consider $\amap$ as the encoding which sends each symbol in $\alphabet$ to the usual binary representation, with $\anumber$ bits, of its associated number in $\dbset{\abs{\alphabet}}$. 

  For each symbol $\padasymbol \in \alphabet \dcup \set{\paddingsymbol}$ and each $h \in \bset{\anumber}$, we let $\amap(\padasymbol,h)$ denote the $h$-th symbol of $\amap(\padasymbol)$, \ie, if $\amap(\padasymbol) = \padbsymbol_{1}\cdots\padbsymbol_{\anumber}$, then $\amap(\padasymbol,h) = \padbsymbol_{h} \in \set{0,1} \dcup \set{\paddingsymbol}$.

  Let $\oddlength^{\prime} = \oddlength \cdot \anumber$, and let $j = g\cdot \anumber + h$, where $g \in \dbset{\oddlength}$ and $h \in \bset{\anumber}$.
  For each $i \in \dbset{\nrelations+1}$, consider the layer $\blayer_{i,j} \in \layeralphabet{\set{0,1}^{\tensorproduct \arityrelation_{i}}}{\width^{2}\cdot\abs{\alphabet}^{\arityrelation_{i}}}$ with left state set: $\layerleftfrontier(\blayer_{i,j}) = \layerleftfrontier(\layer_{i,g+1})$ if $h = 1$, and $\layerleftfrontier(\blayer_{i,j}) = \dbset{\abs{\layertransitions(\layer_{i,g+1})}}$ otherwise; right state set:  $\layerrightfrontier(\blayer_{i,j}) = \layerrightfrontier(\layer_{i,g+1})$ if $h = \anumber$, and $\layerrightfrontier(\blayer_{i,j}) = \dbset{\abs{\layertransitions(\layer_{i,g+1})}}$ otherwise; transition set 
  \begin{equation*}
    \layertransitions(\blayer_{i,j}) = 
    \begin{cases}
      \begin{multlined}[0.56\textwidth]
        \big\{\tuple*{\layerleftstate, \tuple*{\amap(\padasymbol_{1}), \ldots, \amap(\padasymbol_{\arityrelation_{i}})}, \layerrightstate} \setst \\[-2.5ex] \hfil \tuple{\layerleftstate, \tuple{\padasymbol_{1}, \ldots, \padasymbol_{\arityrelation_{i}}}, \layerrightstate} \in \layertransitions(\layer_{i,g+1})\big)
      \end{multlined} & \text{ if } \anumber = 1\text{} \\[2.5ex]

      \begin{multlined}[0.56\textwidth]
        \big\{\tuple*{\layerleftstate, \tuple*{\rchanged{\amap(\padasymbol_{1},1)}, \ldots, \rchanged{\amap(\padasymbol_{\arityrelation_{i}},1)}}, \cmap(\layertransition)} \setst \\[-2.5ex] \hfil \exists\,\layertransition=\tuple{\layerleftstate, \tuple{\padasymbol_{1}, \ldots, \padasymbol_{\arityrelation_{i}}}, \layerrightstate} \in \layertransitions(\layer_{i,g+1})\big)
      \end{multlined} & \text{ if } \anumber \neq 1 \text{ and } h = 1\text{} \\[2.5ex]

      \begin{multlined}[0.56\textwidth]
        \big\{\tuple*{\cmap(\layertransition), \tuple*{\rchanged{\amap(\padasymbol_{1},\anumber)}, \ldots, \rchanged{\amap(\padasymbol_{\arityrelation_{i}},\anumber)}}, \layerrightstate} \setst \\[-2.5ex] \hfil \exists\,\layertransition=\tuple{\layerleftstate, \tuple{\padasymbol_{1}, \ldots, \padasymbol_{\arityrelation_{i}}}, \layerrightstate} \in \layertransitions(\layer_{i,g+1})\big)
      \end{multlined} & \text{ if } \anumber \neq 1 \text{ and } h = \anumber\text{} \\[2.5ex]

      \begin{multlined}[0.56\textwidth]
        \big\{\tuple*{\cmap(\layertransition), \tuple*{\rchanged{\amap(\padasymbol_{1},h)}, \ldots, \rchanged{\amap(\padasymbol_{\arityrelation_{i}},h)}}, \cmap(\layertransition)} \setst \\[-2.5ex] \hfil \exists\,\layertransition=\tuple{\layerleftstate, \tuple{\padasymbol_{1}, \ldots, \padasymbol_{\arityrelation_{i}}}, \layerrightstate} \in \layertransitions(\layer_{i,g+1})\big)
      \end{multlined} & \text{ otherwise, } \\[2.5ex]
    \end{cases}
  \end{equation*}
  where $\cmap \colon \layertransitions(\layer_{i,g+1}) \rightarrow \dbset{\abs{\layertransitions(\layer_{i,g+1})}}$ denotes an arbitrary bijection that sends each transition in $\layertransitions(\layer_{i,g+1})$ to an exclusive number in $\dbset{\abs{\layertransitions(\layer_{i,g+1})}}$; 
  initial state set: $\layerinitialstates(\blayer_{i,j}) = \layerinitialstates(\layer_{i,g+1})$ if $j = 1$, and $\layerinitialstates(\blayer_{i,j}) = \emptyset$ otherwise; final state set: $\layerfinalstates(\blayer_{i,j}) = \layerfinalstates(\layer_{i,g+1})$ if $j = \oddlength^{\prime}$, and $\layerfinalstates(\blayer_{i,j}) = \emptyset$ otherwise; initial Boolean flag: $\layerinitialflag(\blayer_{i,j}) = \true$ if $j = 1$, and $\layerinitialflag(\blayer_{i,j}) = \false$ otherwise; and final Boolean flag: $\layerfinalflag(\blayer_{i,j}) = \true$ if $j = \oddlength^{\prime}$, and $\layerfinalflag(\blayer_{i,j}) = \false$ otherwise.

  One can verify that, for each $i \in \dbset{\nrelations+1}$, $\bodd_{i} = \blayer_{i,1}\cdots\blayer_{i,\oddlength^{\prime}}$ is an ODD in $\odddefiningset{\set{0,1}^{\tensorproduct \arityrelation_{i}}}{\width^{2}\cdot \abs{\alphabet}^{\arityrelation_{i}}}{\oddlength^{\prime}}$ with language  
  \begin{equation*}
    \begin{multlined}[0.95\textwidth]
      \oddlang{\bodd_{i}} = \Big\{\big(\amap(\padasymbol_{1,1}, 1),\ldots,\amap(\padasymbol_{1,\arityrelation_{i}}, 1)\big) \cdots \tuple*{\amap(\padasymbol_{1,1}, \anumber),\ldots,\amap(\padasymbol_{1,\arityrelation_{i}},\anumber)} \cdots \\ \hfill \big(\amap(\padasymbol_{\oddlength,1}, 1),\ldots,\amap(\padasymbol_{\oddlength,\arityrelation_{i}}, 1)\big) \cdots \big(\amap(\padasymbol_{\oddlength,1}, \anumber),\ldots,\amap(\padasymbol_{\oddlength,\arityrelation_{i}}, \anumber)\big) \setst \hspace{4.0ex} \\ \hfill \big(\padasymbol_{1,1},\ldots,\padasymbol_{1,\arityrelation_{i}}\big) \cdots \big(\padasymbol_{\oddlength,1},\ldots,\padasymbol_{\oddlength,\arityrelation_{i}}\big) \in \oddlang{\bodd_i}\Big\}\text{.}
    \end{multlined}
  \end{equation*}

  Let $\bstructure$ be the $\vocabulary$-structure, with domain $\structuredomain(\bstructure) = \associatedrelation{\oddlang{\bodd_{0}}}$, such that $\relationsymbol_{i}(\bstructure) = \associatedrelation{\oddlang{\bodd_{i}}}$ for each $i \in \bset{\nrelations}$.
  By construction, $\bstructure$ is $(\set{0,1}, \width^{2}\cdot\abs{\alphabet}^{\max\set{1,\arityrelation_1,\ldots,\arityrelation_{\nrelations}}})$-decisional.
  Moreover, $\bstructure$ is isomorphic to $\structure$.
  Indeed, note that the permutation $\isomorphism \colon \structuredomain(\bstructure) \rightarrow \structuredomain(\structure)$, where $\permutation(\vertex) = \padasymbol_{1}\cdots\padasymbol_{\oddlength}$ for each $\vertex = \amap(\padasymbol_{1})\cdots\padasymbol(\asymbol_{\oddlength}) \in \structuredomain(\bstructure)$, is an isomorphism from $\bstructure$ to $\structure$.
\end{proof}

\paragraph{\bf The hypercube graph $\hypercubegraph{\oddlength}$.}
Let $\oddlength \in \pN$. The \emph{$\oddlength$-dimensional hypercube graph} is
the graph 
$\hypercubegraph{\oddlength} = 
(\relationsymbol_0(\hypercubegraph{\oddlength}),\relationsymbol_1(\hypercubegraph{\oddlength}))$ 
whose vertex set $\relationsymbol_0(\hypercubegraph{\oddlength}) = \{0,1\}^{\oddlength}$ is the set of all 
$\oddlength$-bit binary strings, and whose edge set $\relationsymbol_1(\hypercubegraph{\oddlength})$ 
is the set of all pairs of $\oddlength$-bit strings that differ in exactly one position. 
$$\relationsymbol_1(\hypercubegraph{\oddlength})= \{(b_1b_2\cdots b_{\oddlength}, b_1'b_2'\cdots b_{\oddlength}')\;:\; 
\mbox{There is a unique $i\in [\oddlength]$ such that $b_i\neq b'_i$}\}.$$ 
Therefore, the $\oddlength$-dimensional hypercube is equal to the structure derived from the 
structural pair $\structuraltuple_{\oddlength} = (\odd_0,\odd_1)$, where the ODD $\odd_0 \in \odddefiningset{\{0,1\}}{2}{\oddlength}$ 
accepts all $\oddlength$-bit strings of length $\oddlength$, and $\odd_1\in \odddefiningset{\{0,1\}^{\tensorproduct 2}}{2}{\oddlength}$ 
is the ODD that accepts all strings of the form $b_1\cdots b_k \tensorproduct b_1'\cdots b_k'$ such that $b_1\cdots b_k$ and 
$b_1'\cdots b_k'$ differ in exactly one entry. In other words, we have 
$\hypercubegraph{\oddlength} = \derivedstructure(\structuraltuple_{\oddlength})$. This structural pair
is depicted in Figure \ref{figure:HypercubeODD}.

\begin{figure}[ht]\centering
{\includegraphics[scale=1.7]{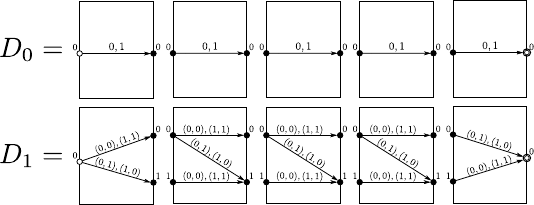}}
\caption{An example of a $(\{0,1\},2,\vocabulary)$-structural pair $\structuraltuple_5 = (\odd_0,\odd_1)$, where $\vocabulary$
is the relational vocabulary of directed graphs. 
The ODD $\odd_0$ accepts all binary strings of length $5$. The odd $\odd_1$ accepts all strings of the 
form $b_1\cdots b_5 \tensorproduct b_1'\cdots b_5' = (b_1,b_1')\cdots (b_5,b_5')$ 
such that $b_1\cdots b_5$ and $b_1'\cdots b_5'$ differ in exactly one bit. 
The structure derived from $\structuraltuple_5$ is the $5$-dimensional directed hypercube 
  $\hypercubegraph{5} = \derivedstructure(\structuraltuple_5)$. \label{figure:HypercubeODD} } 
\end{figure}

\subsection{Regular-Decisional Classes}
\label{subsection:RegularDecisional}
Let $\alphabet$ be an alphabet, $\width \in \pN$ and $\vocabulary = \tuple{\relationsymbol_{1}, \ldots, \relationsymbol_{\nrelations}}$ be a relational vocabulary.
We let $\derivedrelation(\alphabet,\width,\vocabulary)$ 
denote the relation constituted by all $(\alphabet,\width,\vocabulary)$-structural tuples. 
Note that the relation $\derivedrelation(\alphabet,\width,\vocabulary)$ is a subset of the following set of tuples of ODDs 
$$ \bigcup_{\oddlength\in \pN} 
\odddefiningset{\alphabet}{\width}{\oddlength}\cartesianproduct \odddefiningset{\alphabet^{\tensorproduct \arity_1}}{\width}{\oddlength} \cartesianproduct \cdots \cartesianproduct \odddefiningset{\alphabet^{\tensorproduct \arity_{\nrelations}}}{\width}{\oddlength}.$$

Since each $(\alphabet,\width,\vocabulary)$-structural tuple $\structuraltuple\in \derivedrelation(\alphabet,\width,\vocabulary)$ 
corresponds to a $(\alphabet,\width)$-decisional structure $\derivedstructure(\structuraltuple)$, we can associate
with each sub-relation 
$\relation\subseteq \derivedrelation(\alphabet,\width,\vocabulary)$ a class 
$\derivedstructure(\relation) =
\set*{\structure \setst  \structuraltuple \in \relation, \derivedstructure(\structuraltuple) \isomorphic \structure}$
of $(\alphabet,\width)$-decisional $\vocabulary$-structures. 

\begin{definition}
\label{definition:RegularDecisional}
Let $\alphabet$ be an alphabet, $\width\in \pN$, and $\vocabulary$ be a relational vocabulary. 
We say that a class $\classstructures$ of \emph{finite} $\vocabulary$-structures is \emph{$(\alphabet,\width)$-regular-decisional} if 
there exists a regular sub-relation $\relation \subseteq \derivedrelation(\alphabet,\width,\vocabulary)$ 
such that $\classstructures = \derivedstructure(\relation)$. 
\end{definition}

The \emph{$\alphabet$-regular-decisional width} of a class $\classstructures$ of finite $\vocabulary$-structures, denoted by 
$\dwidth{\alphabet}{\classstructures}$, is defined as the minimum $\width \in \pN$ such that 
$\classstructures$ is $(\alphabet,\width)$-regular-decisional. We note that this minimum $w$ may not exist. In this case, 
we set $\dwidth{\alphabet}{\classstructures} = \infty$. 

Now, consider the alphabet 
\begin{equation}
\label{equation:VocabularyAlphabet}
\layervocabularyalphabet{\alphabet}{\width}{\vocabulary} \defeq \layeralphabet{\alphabet}{\width} \tensorproduct \layeralphabet{\alphabet^{\tensorproduct\arity_1}}{\width} \tensorproduct \cdots \tensorproduct \layeralphabet{\alphabet^{\tensorproduct \arity_{\nrelations}}}{\width}.
\end{equation}

\todo[inline]{(Up to now it is not known that the set 
$\layervocabularyalphabet{\alphabet}{\width}{\vocabulary}^{\circledast}$ is regular).
We let $\layervocabularyalphabet{\alphabet}{\width}{\vocabulary}^{\circledast}$
be the set of strings $\odd_{0}\tensorproduct \odd_{1}\tensorproduct \ldots \tensorproduct \odd_{\nrelations}$ over 
$\layervocabularyalphabet{\alphabet}{\width}{\vocabulary}$ such that $(\odd_0,\odd_1,\ldots,\odd_{\nrelations})$ is a 
$(\alphabet,\width,\vocabulary)$-structural tuple.}

Then a class $\classstructures$ of finite $\vocabulary$-structures is $(\alphabet,\width)$-regular-decisional if and only if 
there exists a finite automaton $\finiteautomaton$ over $\layervocabularyalphabet{\alphabet}{\width}{\vocabulary}$ 
such that 
$\associatedrelation{\automatonlang{\finiteautomaton}}\subseteq \derivedrelation(\alphabet,\width,\vocabulary)$, 
and $\classstructures=\derivedstructure(\associatedrelation{\automatonlang{\finiteautomaton}})$. 

\begin{definition}
\label{definition:FiniteAutomatonRepresentation}
Let $\alphabet$ be an alphabet, $\width\in \pN$, $\vocabulary=\tuple{\relationsymbol_1,\dots,\relationsymbol_{\nrelations}}$
be a relational vocabulary and $\classstructures$ be a $(\alphabet,\width,\vocabulary)$-regular-decisional
class of finite structures. We say that a finite automaton 
$\finiteautomaton$ over $\layervocabularyalphabet{\alphabet}{\width}{\vocabulary}$
represents $\classstructures$ if 
$\associatedrelation{\automatonlang{\finiteautomaton}}\subseteq \derivedrelation(\alphabet,\width,\vocabulary)$
and $\classstructures=\derivedstructure(\associatedrelation{\automatonlang{\finiteautomaton}})$. 
\end{definition}

\subsection{The Hypercube Language.}

\todo[inline]{Revisit this paragraph. It is not clear whether we are talking about equality or isomorphism here.
Same thing with definition above.}

In order to show that a class $\classstructures$ of finite $\vocabulary$-structures 
is $(\alphabet,\width)$-decisional, it is enough to construct a finite automaton 
$\finiteautomaton$ over $\layervocabularyalphabet{\alphabet}{\width}{\vocabulary}$
representing $\classstructures$. More precisely, we need to define such an automaton 
$\finiteautomaton$ such that for each $\oddlength\in \N$, and each string 
$\odd_0\tensorproduct \odd_1\tensorproduct \ldots \tensorproduct \odd_{\nrelations}
\in \layervocabularyalphabet{\alphabet}{\width}{\vocabulary}^{\oddlength}$, 
$\odd_0\tensorproduct \odd_1\tensorproduct \ldots \tensorproduct \odd_{\nrelations}$
belongs to $\automatonlang{\finiteautomaton}$ if and only if
$\structuraltuple = (\odd_0,\odd_1,\ldots,\odd_{\nrelations})$ is a
$(\alphabet,\width,\vocabulary)$-structural tuple and $\derivedstructure(\structuraltuple)\in \classstructures$.

To illustrate this type of construction, we show in Proposition
\ref{proposition:hypercube_graph} that the class $\hypercubegraphfamily$ of
hypercubes is $(\set{0,1},2)$-regular-decisional. Since it can be easily shown
that this class is not $(\set{0,1},1)$-regular-decisional, we have that the
$\set{0,1}$-regular-decisional width of $\hypercubegraphfamily$ is $2$.

\begin{proposition}
\label{proposition:hypercube_graph}
Let $\hypercubegraphfamily \defeq \set*{\hypercubegraph{\nbits} \setst \nbits \in \pN}$ be the class of all hypercube graphs.
The class $\hypercubegraphfamily$ is $(\set{0,1}, 2)$-regular-decisional.
\end{proposition}
\begin{proof}
In Figure \ref{figure:HypercubeClass} we depict an automaton $\finiteautomaton$ which accepts a 
string $\odd_0\tensorproduct \odd_1$ of length $\oddlength$ if and only if the pair 
$\structuraltuple_{\oddlength} = (\odd_0,\odd_1)$ is a structural pair whose derived structure 
  $\derivedstructure(\structuraltuple_{\oddlength})$ is the hypercube graph $\hypercubegraph{k}$. 
\end{proof}

\begin{figure}[ht]\centering
\includegraphics[scale=1.2]{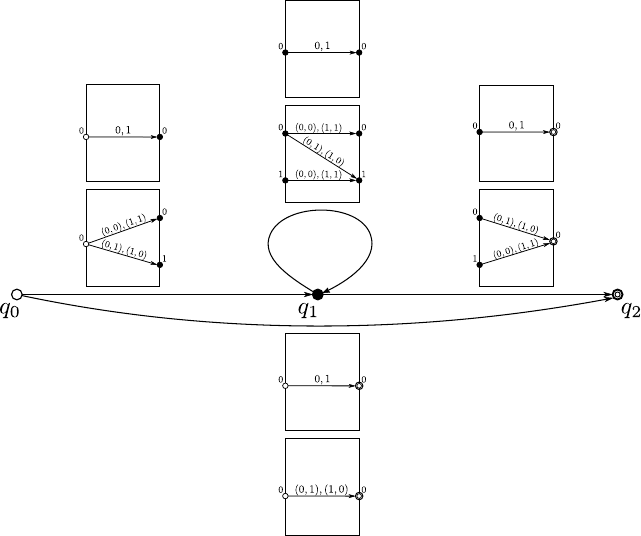}
\caption{An automaton $\finiteautomaton$ over the alphabet $\layervocabularyalphabet{\{0,1\}}{2}{\vocabulary}$,  
where $\vocabulary$ is the vocabulary of directed graphs. 
This automaton accepts exactly one string of length $\oddlength$ for each $\oddlength \in \pN$. 
For each such $\oddlength$, if $\odd_0\tensorproduct \odd_1$ is the unique string of length $\oddlength$ accepted 
by $\finiteautomaton$, then the pair $\structuraltuple_{\oddlength} = (\odd_0,\odd_1)$ is structural, 
and $\derivedstructure(\structuraltuple_{\oddlength})$ is the hypercube graph $\hypercubegraph{\oddlength}$.
In particular, the string $\odd_0\tensorproduct \odd_1$ represented in Figure \ref{figure:HypercubeODD} is accepted 
by $\finiteautomaton$ upon following the sequence of states $\state_0\state_1\state_1\state_1\state_1\state_2$.}\label{figure:HypercubeClass}
\end{figure}

An interesting aspect of Proposition \ref{proposition:hypercube_graph} is 
that it states that the class $\hypercubefamily$ of hypercube graphs has 
regular-decisional width $2$, 
while this class has unbounded width with respect to most 
traditional width measures studied in structural graph theory. For instance, it can be shown that the hypercube 
graph $\hypercubegraph{\oddlength}$ has treewidth $\Theta(2^{\oddlength}/\sqrt{\oddlength})$ 
\cite{Chandran2006treewidth} and cliquewidth $\Omega(2^{\oddlength}/\sqrt{\oddlength})$ 
\cite{bonomo2016graph}. Therefore, $\hypercubegraph{\oddlength}$ has also exponential
bandwidth, carving width, pathwidth, treedepth and rank-width. Additionally,
since all vertices of $\hypercubegraph{\oddlength}$ have degree $\oddlength$, 
the degeneracy of the family $\hypercubefamily$ is $\Theta(\oddlength)$. 
Therefore, $\hypercubefamily$ is {\em not} a nowhere dense class of graphs. 
This also implies that the graphs in $\hypercubefamily$ have unbounded genus, unbounded local treewidth, etc.

\section{First-Order Definable Classes of Structures of Constant Width}
\label{section:FirstOrderLogic}

For each $\fol{\vocabulary}$-sentence $\formula$, we let
$\derivedrelation(\alphabet,\width,\vocabulary,\formula)$ denote the
subset of $\derivedrelation(\alphabet,\width,\vocabulary)$ consisting of
all tuples $\structuraltuple \in
\derivedrelation(\alphabet,\width,\vocabulary)$ whose associated structure
$\derivedstructure(\structuraltuple)$ satisfies $\formula$.
$$\derivedrelation(\alphabet,\width,\vocabulary,\formula) \defeq
\set{\structuraltuple \in \derivedrelation(\alphabet,\width,\vocabulary) \setst
\derivedstructure(\structuraltuple) \models \formula}\text{.}$$ The next theorem
(Theorem~\ref{theorem:FirstOrderRegular}), states that the relation
$\derivedrelation(\alphabet,\width,\vocabulary,\formula)$ is regular. 
 
\begin{theorem}\label{theorem:FirstOrderRegular}
Let $\alphabet$ be an alphabet, $\width\in \pN$, and $\vocabulary = \tuple{\relationsymbol_{1},\ldots,\relationsymbol_{\nrelations}}$
be a relational vocabulary. For each $\fol{\vocabulary}$-sentence $\formula$, the relation
$\derivedrelation(\alphabet,\width,\vocabulary,\formula)$ is regular. 
\end{theorem}

A constructive proof of Theorem \ref{theorem:FirstOrderRegular} will be given in Section~\ref{section:ProofMainTheorem}. 
As a consequence of Theorem~\ref{theorem:FirstOrderRegular}, we have that the problem of determining whether a given
$\fol{\vocabulary}$-sentence $\formula$ is satisfied  by some structure $\astructure$ in the class of structures specified
by a given finite automaton $\finiteautomaton$ over the alphabet $\layervocabularyalphabet{\alphabet}{\width}{\vocabulary}$ is 
fixed-parameter linear in the size of $\finiteautomaton$ when parameterized by $\vocabulary$, $\formula$, $\alphabet$
and $\width$ (Theorem \ref{theorem:DecidableFirstOrderTheory}). 

We note that in natural applications, the parameter $\vocabulary$ is often fixed.
For instance, if our goal is to decide properties about classes of directed graphs, then 
$\vocabulary$ is simply the vocabulary of directed graphs. Additionally, often the sentence 
$\formula$ is fixed, as in the case where the goal is to determine whether all graphs in
a given class of graphs, specified by a finite automaton $\finiteautomaton$
over the alphabet $\layervocabularyalphabet{\alphabet}{\width}{\vocabulary}$,
are triangle-free. Finally, by Proposition \ref{proposition:odd_alphabet},
the alphabet $\alphabet$ can be assumed to be the binary alphabet $\{0,1\}$, with only a moderate increase in width. Therefore, 
in applications where $\vocabulary$, $\formula$ and $\alphabet$ are fixed, the only relevant
parameter is the width parameter $\width$. This parameter intuitively is a measure of 
the complexity of each individual structure in the class specified by $\finiteautomaton$, while
the size of $\finiteautomaton$ is a measure of the complexity of the class itself.

\begin{theorem}\label{theorem:DecidableFirstOrderTheory}
Let $\alphabet$ be an alphabet, $\width\in \pN$, $\vocabulary = \tuple{\relationsymbol_1,\dots,\relationsymbol_{\nrelations}}$
be a relational vocabulary, and $\formula$ be a $\fol{\vocabulary}$-sentence. Given a finite automaton
$\finiteautomaton$ over the alphabet $\layervocabularyalphabet{\alphabet}{\width}{\vocabulary}$
representing a $(\alphabet,\width)$-regular-decisional class of $\vocabulary$-structures,
one can determine in time $f(\alphabet,\width,\vocabulary,\formula) \cdot |\finiteautomaton|$
whether some $\vocabulary$-structure $\astructure\in \classstructures$ satisfies $\formula$, for some 
computable function $f$.
\end{theorem}
\begin{proof}
The proof of Theorem~\ref{theorem:FirstOrderRegular} shows how to construct 
an automaton $\finiteautomaton(\alphabet,\width,\vocabulary,\formula)$ 
over the alphabet $\layervocabularyalphabet{\alphabet}{\width}{\vocabulary}$
such that
$\automatonlang{\finiteautomaton(\alphabet,\width,\vocabulary,\formula)}
= \associatedlang{\derivedrelation(\alphabet,\width,\vocabulary,\formula)}$.
Therefore, there is some $\vocabulary$-structure $\astructure\in \classstructures$
that satisfies $\formula$, if and only if the language 
$\automatonlang{\finiteautomaton}\cap \automatonlang{\finiteautomaton(\alphabet,\width,\vocabulary,\formula)}$
is non-empty. Since this non-emptiness of intersection test can be performed in time
$O(|\finiteautomaton(\alphabet,\width,\vocabulary,\formula)| \cdot |\finiteautomaton|)$, the 
theorem follows by setting $f(\alphabet,\width,\vocabulary,\formula) = 
c\cdot |\finiteautomaton(\alphabet,\width,\vocabulary,\formula)| + d$ for some sufficiently large 
constants $c,d\in \N$. 
\end{proof}

\newcommand{\anotheralphabet}{\Gamma}

Let $\vocabulary = \tuple{\relationsymbol_{1}, \ldots, \relationsymbol_{\nrelations}}$ be a relational vocabulary.  We denote by $\msol{\vocabulary}$ the set
of all monadic second-order logic formulas over $\vocabulary$, \ie the
extension of $\fol{\vocabulary}$ that, additionally, allows variables to be
used as placeholders for sets of elements from the domain of a
finite $\vocabulary$-structure and allows quantification over such variables. 
We note that neither Theorem \ref{theorem:FirstOrderRegular} nor Theorem
\ref{theorem:DecidableFirstOrderTheory} can be generalized to the logic
$\msol{\vocabulary}$ for an arbitrary relational vocabulary $\vocabulary$.
Indeed, it is well known that the MSO theory of unlabeled grids is already
undecidable~\cite{Seese1991structure,Hlinveny2006}
\footnote{Note that the well known fact that the first-order theory of {\em unlabeled} grids
is {\em decidable} is a special case of Theorem \ref{theorem:DecidableFirstOrderTheory},
since unlabeled grids have constant decisional-width. On the other hand, it is also 
well known that the first-order theory of {\em labeled grids} is {\em undecidable}. 
Note that labeled grids may have arbitrarily large decidisonal-width, 
due to the fact that the ODDs representing
the vertices in each label class may have arbitrarily large width.}

Nevertheless, instead of using MSO logic to reason about properties of classes of
$\vocabulary$-structures, we can use MSO logic over the vocabulary of 
$\layervocabularyalphabet{\alphabet}{\width}{\vocabulary}$-strings to
{\em define} $(\alphabet,\width)$-regular-decisional classes of $\vocabulary$-structures. 
More precisely, for a given alphabet $\anotheralphabet$, let $\svocabulary[\anotheralphabet]$ be the
vocabulary of strings over $\anotheralphabet$. From B\"{u}chi-Elgot's theorem 
\cite{Buchi1960,Elgot1961}, a language $\lang\subseteq \anotheralphabet^*$ is regular if and only
if $\lang$ can be defined by an $\msol{\svocabulary[\anotheralphabet]}$-sentence.
In particular, a class $\classstructures$ of finite $\vocabulary$-structures is $(\alphabet,\width)$-regular decisional if and only if there is an
$\msol{\svocabulary[\layervocabularyalphabet{\alphabet}{\width}{\vocabulary}]}$-sentence $\MSOformula$
such that for each $\oddlength\in \pN$, and each string 
$\structurestring= \odd_{0}\tensorproduct \odd_{1}\tensorproduct \ldots \tensorproduct \odd_{\nrelations}$ in
$\layervocabularyalphabet{\alphabet}{\width}{\vocabulary}^{\oddlength}$, $\structurestring$ satisfies $\MSOformula$ if and only if
$\structuraltuple = (\odd_0,\odd_1,\ldots,\odd_{\nrelations})$ is a $(\alphabet,\width,\vocabulary)$-structural tuple
and $\derivedstructure(\structuraltuple)$ belongs to $\classstructures$. 
 
\begin{theorem} 
\label{theorem:StringsVsStructures}
  Let $\alphabet$ be an alphabet, $\width \in \pN$ and $\vocabulary = \tuple{\relationsymbol_{1}, \ldots, \relationsymbol_{\nrelations}}$ be a relational vocabulary.
  Given an $\msol{\svocabulary[\layervocabularyalphabet{\alphabet}{\width}{\vocabulary}]}$-sentence
  $\MSOformula$ and an $\fol{\vocabulary}$-sentence $\formula$, one can 
  decide whether there exists some string $\structurestring = \odd_{0}\tensorproduct \odd_{1}\tensorproduct \ldots \tensorproduct \odd_{\nrelations} \in \playerstructuredefiningset{\alphabet}{\width}{\vocabulary}$ 
  such that $\structurestring\models \MSOformula$ and $\derivedstructure(\structuraltuple) \models \formula$, where $\structuraltuple = \tuple{\odd_{0}, \odd_{1}, \ldots, \odd_{\nrelations}}$.
\end{theorem}
\begin{proof} 
  By using B\"{u}chi-Elgot's Theorem, one can construct a finite automaton $\finiteautomaton_{1}$ over $\layervocabularyalphabet{\alphabet}{\width}{\vocabulary}$ that 
  accepts a string \anote{(R2.31)}$\structurestring\in \layervocabularyalphabet{\alphabet}{\width}{\vocabulary}^{+}$
  if and only if $\structurestring\models \MSOformula$. Now, from Theorem \ref{theorem:FirstOrderRegular},
        we can construct a finite automaton $\finiteautomaton_{2}$ 
  over $\layervocabularyalphabet{\alphabet}{\width}{\vocabulary}$ which accepts a string 
  $\odd_{0}\tensorproduct \odd_{1}\tensorproduct \ldots \tensorproduct \odd_{\nrelations} \in 
  \layervocabularyalphabet{\alphabet}{\width}{\vocabulary}^{+}$ if and only if 
  $\structuraltuple = \tuple{\odd_0,\odd_1,\ldots,\odd_{\nrelations}}$ is a $(\alphabet,\width,\vocabulary)$-structural tuple and 
  $\derivedstructure(\structuraltuple) \models \formula$. Let $\finiteautomaton_{\cap}$ be a finite automaton
  that accepts the language $\automatonlang{\finiteautomaton_1}\cap \automatonlang{\finiteautomaton_2}$. Then, we have that 
  $\automatonlang{\finiteautomaton_{\cap}}$ is non-empty if and only if there exists some string $\structurestring = \odd_{0}\tensorproduct \odd_{1}\tensorproduct \ldots \tensorproduct \odd_{\nrelations} \in \playerstructuredefiningset{\alphabet}{\width}{\vocabulary}$ 
  such that $\structurestring\models \MSOformula$ and $\derivedstructure(\structuraltuple) \models \formula$, where $\structuraltuple = \tuple{\odd_{0}, \odd_{1}, \ldots, \odd_{\nrelations}}$. 
  Since emptiness is decidable for finite automata, the theorem follows. 
\end{proof}

\section{Proof of Theorem~\ref{theorem:FirstOrderRegular}}\label{section:ProofMainTheorem}
We dedicate this section to the proof of Theorem~\ref{theorem:FirstOrderRegular}.
The proof follows a traditional strategy combined with new machinery for the implicit
manipulation of ODDs. More precisely, given an alphabet $\alphabet$, $\width \in \pN$, a relational vocabulary 
$\vocabulary = \tuple{\relationsymbol_{1},\ldots,\relationsymbol_{\nrelations}}$ and an 
$\fol{\vocabulary}$-formula $\formula(\variable_1,\ldots,\variable_{\nfreevariables})$ with 
free variables $\freevariables(\formula) \subseteq\, \freevariableset_{\nfreevariables} = \{\variable_1,\ldots,\variable_{\nfreevariables}\}$, 
we define $\derivedrelation(\alphabet,\width,\vocabulary,\formula,\freevariableset_{\nfreevariables})$ as the relation 
containing precisely the tuples of the form $(\odd_0,\odd_1,\ldots,\odd_{\nrelations},\vertex_1,\ldots,\vertex_{\nfreevariables})$ 
such that $\structuraltuple = (\odd_0,\odd_1,\ldots,\odd_{\nrelations})$ is $(\alphabet,\width,\vocabulary)$-structural 
and $\derivedstructure(\structuraltuple) \models \formula[\vertex_1,\ldots,\vertex_{\nfreevariables}]$. 
The Boolean connectives $\wedge,\vee$ and $\neg$ and the existential 
quantification $\exists$ are handled using closure properties from regular languages. 

The technically involved part of the proof however will be the construction of an {\em initial} automaton which accepts precisely those
strings $$\odd_0\tensorproduct \odd_1\tensorproduct \ldots \tensorproduct \odd_{\nrelations}\tensorproduct \vertex_1\ldots\vertex_{\nfreevariables}$$
such that $(\odd_0,\odd_1,\ldots,\odd_{\nrelations})$ is a $(\alphabet,\width,\vocabulary)$-structural tuple,
and $\vertex_1,\ldots,\vertex_{\nfreevariables}$ belong to the domain $\oddlang{\odd_0}$. 
In particular, we need to guaranteed that for each $i\in [\nrelations]$, 
the language $\oddlang{\odd_i}$ is contained in the tensored language $\oddlang{\odd_0}^{\tensorproduct \arity_i}$.

\subsection{Basic General Operations}
In this section, we introduce some basic low-level operations that will be used repeatedly in 
the proof of Theorem~\ref{theorem:FirstOrderRegular}. More precisely, we consider the following
operations: \emph{projection}, \emph{identification}, \emph{permutation of coordinates}, \emph{fold}, \emph{unfold}, \emph{direct sum}, \emph{union}, \emph{intersection} and \emph{complementation}.

Let $\alphabet$ be an alphabet, $\arityrelation \in \pN$ and $\relation \subseteq (\alphabet^{+})^{{\cartesianproduct} \arityrelation}$ be an $\arityrelation$-ary relation.

For each permutation $\permutation \colon \bset{\arityrelation} \rightarrow \bset{\arityrelation}$, we let $\perm(\relation,\permutation)$ be the relation obtained from $\relation$ by permuting the coordinates of each tuple in $\relation$ according to $\permutation$. In other words, $\perm(\relation,\permutation) \defeq \set*{\tuple*{\astring_{\permutation(1)},\ldots,\astring_{\permutation(\arityrelation)}} \setst \tuple*{\astring_{1}, \ldots, \astring_{\arityrelation}} \in \relation}$.

For each $i \in \bset{\arityrelation}$, the \emph{projection of the $i$-th coordinate of $\relation$} is defined as the $(\arityrelation-1)$-ary relation $\projection(\relation,i) \defeq \set*{\tuple*{\astring_1,\dots,\astring_{i-1},\astring_{i+1},\dots,\astring_{\arityrelation}} \setst \tuple*{\astring_{1}, \ldots, \astring_{\arityrelation}} \in \relation}\text{}$ 
obtained from $\relation$ by removing the $i$-th coordinate of each tuple in $\relation$.
More generally, for each $\indexset \subseteq \bset{\arityrelation}$, we let $\projection(\relation,\indexset)$ denote the relation obtained from $\relation$ by removing all the $i$-th coordinates of each tuple in $\relation$, where $i \in \indexset$.

For each $i, j \in \bset{\arityrelation}$, the \emph{identification of the $i$-th and $j$-th coordinates of $\relation$} is defined as the relation $\identify(\relation,i,j) \defeq \set{\tuple{\astring_{1},\ldots,\astring_{\arityrelation}} \in \relation \setst \astring_{i}=\astring_{j}}\text{}$ obtained from $\relation$ by removing each tuple $\tuple{\astring_{1},\ldots,\astring_{\arityrelation}} \in \relation$ such that $\astring_{i} \neq \astring_{j}$.
More generally, for each $\indexset \subseteq \bset{\arityrelation} \cartesianproduct \bset{\arityrelation}$, we let $\identify(\relation,\indexset) \defeq \set{\tuple{\astring_{1},\ldots,\astring_{\arityrelation}} \in \relation \setst \astring_{i}=\astring_{j}, \tuple{i,j} \in \indexset}$.

For each $i, j \in \bset{\arityrelation}$, with $i \leq j$, we let $$\fold(\relation,i,j) \defeq \set*{\tuple*{\astring_{1}, \ldots, \astring_{i-1}, \astring_{i} \tensorproduct \cdots \tensorproduct \astring_{j}, \astring_{j+1}, \ldots, \astring_{\arityrelation}} \setst \tuple*{\astring_{1}, \ldots, \astring_{\arityrelation}} \in \relation}\text{.}$$
On the other hand, if $\arelation = \fold(\brelation,i,j)$ for some relation $\brelation$ and some $i, j \in \bset{\arityrelation}$, with $i \leq j$, then we let $\unfold(\arelation,i) = \brelation$, \ie the inverse operation of $\fold$.

Let $\aalphabet$ and $\balphabet$ be two alphabets, $\arityrelation_{1}, \arityrelation_{2} \in \pN$, $\relation_{1} \subseteq (\aalphabet^{+})^{\cartesianproduct \arityrelation_{1}}$ be an $\arityrelation_{1}$-ary relation and $\relation_{2} \subseteq (\balphabet^{+})^{{\cartesianproduct} \arityrelation_{2}}$ be an $\arityrelation_{2}$-ary relation. 
If $\relation_{1}$ and $\relation_{2}$ are non-empty, then we define the \emph{direct sum} of $\relation_{1}$ with $\relation_{2}$ as the $(\arityrelation_{1}+\arityrelation_{2})$-ary relation 
\begin{equation*}
  \relation_{1}\directsum \relation_{2} \defeq \set*{\tuple*{\astring_{1}, \ldots, \astring_{\arityrelation_{1}}, \bstring_{1}, \ldots, \bstring_{\arityrelation_{2}}} \setst \tuple*{\astring_{1}, \ldots, \astring_{\arityrelation_{1}}} \in \relation_{1}, \tuple*{\bstring_{1}, \ldots, \bstring_{\arityrelation_{2}}} \in \relation_{2}}\text{.}
\end{equation*}
Otherwise, we let $\relation_{1} \directsum \emptyset \defeq \relation_{1}$ and $\emptyset \directsum \relation_{2} \defeq \relation_{2}$.

\begin{proposition}\label{proposition:BasicProperties1}
Let $\aalphabet$ and $\balphabet$ be two alphabets, $\arityrelation_{1}, \arityrelation_{2} \in \pN$, $\relation_{1} \subseteq (\aalphabet^{+})^{{\cartesianproduct} \arityrelation_{1}}$ be a regular $\arityrelation_{1}$-ary relation and $\relation_{2} \subseteq (\balphabet^{+})^{{\cartesianproduct} \arityrelation_{2}}$ be a regular $\arityrelation_{2}$-ary relation. 
  The following closure properties are held: 
  \begin{enumerate}
    \item\label{item:perm} for each permutation $\permutation \colon \bset{\arityrelation_{1}} \rightarrow \bset{\arityrelation_{1}}$, $\perm(\relation_{1},\permutation)$ is regular;
    \item\label{item:proj} for each $\indexset \subseteq \bset{\arityrelation_{1}}$, $\projection(\relation_{1},\indexset)$ is regular;
    \item\label{item:ident} for each $\indexset \subseteq \bset{\arityrelation_{1}} \times \bset{\arityrelation_{1}}$, $\identify(\relation_{1},\indexset)$ is regular;
    \item\label{item:fold} for each $i, j \in \bset{\arityrelation_{1}}$, with $i \leq j$, $\fold(\relation_{1},i,j)$ is regular;
    \item\label{item:unfold} if $\arelation_{1} = \fold(\brelation_{1},i,j)$ for some relation $\brelation_{1}$ and some $i, j \in \bset{\arityrelation_{1}}$, with $i \leq j$, then $\unfold(\relation_{1},i) = \brelation_{1}$ is regular;
    \item\label{item:directsum} $\relation_{1}\directsum \relation_{2}$ is regular.
  \end{enumerate}
\end{proposition}
\begin{proof}
  Since $\relation_{1}$ is a regular relation, there exists a finite automaton $\finiteautomaton_{1}$ over the alphabet $\aalphabet^{\tensorproduct \arityrelation_{1}}$ such that $\automatonlang{\finiteautomaton_{1}} = \associatedlang{\relation_{1}}$.
  Analogously, there exists a finite automaton $\finiteautomaton_{2}$ over the alphabet $\balphabet^{\tensorproduct \arityrelation_{2}}$ such that $\automatonlang{\finiteautomaton_{2}} = \associatedlang{\brelation}$. 

  \ref{item:perm}. Let $\finiteautomaton_{1}^{\prime}$ be the finite automaton with state set $\automatonstates(\finiteautomaton_{1}^{\prime}) = \automatonstates(\finiteautomaton_{1})$, transition set $\automatontransitions(\finiteautomaton_{1}^{\prime}) = \set*{\tuple{\bstate, \tuple{\padasymbol_{\permutation(1)}, \ldots, \padasymbol_{\permutation(\arityrelation_{1})}}, \astate} \setst \tuple{\bstate, \tuple{\padasymbol_{1}, \ldots, \padasymbol_{\arityrelation_{1}}}, \astate} \in \automatontransitions(\finiteautomaton_{1})}\text{,}$ initial state set $\automatoninitialstates(\finiteautomaton_{1}^{\prime}) = \automatoninitialstates(\finiteautomaton_{1})$, and final state set $\automatonfinalstates(\finiteautomaton_{1}^{\prime}) = \automatonfinalstates(\finiteautomaton_{1})$.
  One can verify that $\automatonlang{\finiteautomaton_{1}^{\prime}} = \associatedlang{\perm(\relation_{1},\permutation)}$.
  Therefore, $\perm(\relation_{1},\permutation)$ is regular.
   
  \ref{item:proj}. For each symbol $\tuple{\padasymbol_{1}, \ldots, \padasymbol_{\arityrelation_{1}}} \in (\alphabet_{1} \dcup \set{\paddingsymbol})^{\tensorproduct \arityrelation_{1}}$, we let $\projection(\padasymbol_{1}, \ldots, \padasymbol_{\arityrelation_{1}}, \indexset)$ denote the symbol in $(\alphabet_{1} \dcup \set{\paddingsymbol})^{\tensorproduct (\arityrelation_{1}-\abs{\indexset})}$ obtained from $\tuple{\padasymbol_{1}, \ldots, \padasymbol_{\arityrelation_{1}}}$ by removing each $i$-th coordinate, where $i \in \indexset$.
  Thus, let $\finiteautomaton_{1}^{\prime}$ be the finite automaton with state set $\automatonstates(\finiteautomaton_{1}^{\prime}) = \automatonstates(\finiteautomaton_{1})$, transition set $$\automatontransitions(\finiteautomaton_{1}^{\prime}) = \set*{\tuple{\bstate, \projection(\padasymbol_{1}, \ldots, \padasymbol_{\arityrelation_{1}}, \indexset), \astate} \in \automatontransitions(\finiteautomaton_{1}) \setst \tuple{\bstate, \tuple{\padasymbol_{1}, \ldots, \padasymbol_{\arityrelation_{1}}}, \astate} \in \automatontransitions(\finiteautomaton_{1})}\text{,}$$ initial state set $\automatoninitialstates(\finiteautomaton_{1}^{\prime}) = \automatoninitialstates(\finiteautomaton_{1})$, and final state set $\automatonfinalstates(\finiteautomaton_{1}^{\prime}) = \automatonfinalstates(\finiteautomaton_{1})$.
  One can verify that $\automatonlang{\finiteautomaton_{1}^{\prime}} = \associatedlang{\projection(\relation_{1},\indexset)}$.
  Therefore, $\projection(\relation_{1},\indexset)$ is regular.

  \ref{item:ident}. Let $\finiteautomaton_{1}^{\prime}$ be the finite automaton with state set $\automatonstates(\finiteautomaton_{1}^{\prime}) = \automatonstates(\finiteautomaton_{1})$, transition set $\automatontransitions(\finiteautomaton_{1}^{\prime}) = \automatontransitions(\finiteautomaton_{1}) \setminus \set*{\tuple{\bstate, \tuple{\padasymbol_{1}, \ldots, \padasymbol_{\arityrelation_{1}}}, \astate} \in \automatontransitions(\finiteautomaton_{1}) \setst \padasymbol_{i} \neq \padasymbol_{j}, \tuple{i,j} \in \indexset}\text{,}$ initial state set $\automatoninitialstates(\finiteautomaton_{1}^{\prime}) = \automatoninitialstates(\finiteautomaton_{1})$, and final state set $\automatonfinalstates(\finiteautomaton_{1}^{\prime}) = \automatonfinalstates(\finiteautomaton_{1})$.
  One can verify that $\automatonlang{\finiteautomaton_{1}^{\prime}} = \associatedlang{\identify(\relation_{1},\indexset)}$.
  Therefore, $\identify(\relation_{1},\indexset)$ is regular.

\ref{item:fold}. Let $\alpha$ be the map that sends each tuple of symbols $(\sigma_1,...,\sigma_{\arity_1})\in \alphabet_1^{\arity_1}$ to the tuple 
	$(\sigma_1,...,\sigma_{i-1},(\sigma_i,...,\sigma_j ),\sigma_{j+1},...,\sigma_{\arity_1})$. For a language
	$L\subseteq (\alphabet_1^{\otimes \arity_1})^{+}$, let $\alpha(L)$ be the homomorphic image of $L$ under $\alpha$. 
	Since regular languages are closed under homomorphism, and $\fold(\relation_{1},i,j) = \associatedrelation{\alpha(\associatedlang{\relation_1})}$ we have 
	that $\fold(\relation_1,i,j)$ is regular. 

	\ref{item:unfold}. Let $\beta$ be the map that sends each tuple $(\sigma_1,...,\sigma_{i-1},(\sigma_i,...,\sigma_j ),\sigma_{j+1},...,\sigma_{\arity_1})$ 
	to the tuple $(\sigma_1,...,\sigma_{\arity_1})\in \alphabet_1^{\arity_1}$. For a language	
	$L\subseteq \alphabet_1^{\otimes_{i-1}}\times (\alphabet_1^{\otimes j-i+1})\times\alphabet_1^{\otimes \arity_1-j+1}$ we let $\beta(L)$ be 
	the homomorphic image of $L$ under $\beta$. Since regular languages are closed under homomorphism, and 
	$\unfold(\relation_1,i) = \associatedrelation{\beta(\associatedlang{\relation_1})}$, we have that 
	$\unfold(\relation_1,i)$ is regular.  

  \ref{item:directsum}. For each $i \in \bset{2}$, let $\finiteautomaton_{i}^{\prime}$ be the finite automaton with state set $\automatonstates(\finiteautomaton_{i}^{\prime}) = \automatonstates(\finiteautomaton_{i})$, transition set $\automatontransitions(\finiteautomaton_{i}^{\prime}) = \automatontransitions(\finiteautomaton_{i}) \cup \set{\tuple{\state, \paddingtuple{\arityrelation_{i}}, \state} \setst \state \in \automatonfinalstates(\finiteautomaton_{i})}\text{,}$ initial state set $\automatoninitialstates(\finiteautomaton_{i}^{\prime}) = \automatoninitialstates(\finiteautomaton_{i})$ and final state set $\automatonfinalstates(\finiteautomaton_{i}^{\prime}) = \automatonfinalstates(\finiteautomaton_{i})$.
  Now, consider the finite automaton $\finiteautomaton^{\prime}$ with state set $\automatonstates(\finiteautomaton^{\prime}) = \automatonstates(\finiteautomaton_{1}^{\prime}) \cartesianproduct \automatonstates(\finiteautomaton_{2}^{\prime})$, transition set
  \begin{equation*}
    \begin{multlined}[0.9\displaywidth]
      \automatontransitions(\finiteautomaton^{\prime}) = \Big\{\tuple*{\tuple{\bstate_{1},\bstate_{2}}, \tuple{\padasymbol_{1}, \ldots, \padasymbol_{\arityrelation_{1}}, \padbsymbol_{1}, \ldots, \padbsymbol_{\arityrelation_{2}}}, \tuple{\astate_{1},\astate_{2}}} \setst \\[1.0ex] \hfil \tuple*{\bstate_{1}, \tuple{\padasymbol_{1}, \ldots, \padasymbol_{\arityrelation_{1}}}, \astate_{1}} \in \automatontransitions(\finiteautomaton_{1}^{\prime}), \tuple*{\bstate_{2}, \tuple{\padbsymbol_{1}, \ldots, \padbsymbol_{\arityrelation_{2}}}, \astate_{2}} \in \automatontransitions(\finiteautomaton_{2}^{\prime}) \\ \hfil \tuple{\padasymbol_{1}, \ldots, \padasymbol_{\arityrelation_{1}}, \padbsymbol_{1}, \ldots, \padbsymbol_{\arityrelation_{2}}} \neq \paddingtuple{(\arityrelation_{1}+ \arityrelation_{2})}\Big\}\text{,}
    \end{multlined}
  \end{equation*}
  initial state set $\automatoninitialstates(\finiteautomaton^{\prime}) = \automatoninitialstates(\finiteautomaton_{1}^{\prime}) \cartesianproduct \automatoninitialstates(\finiteautomaton_{2}^{\prime})$ and final state set $\automatonfinalstates(\finiteautomaton^{\prime}) = \automatonfinalstates(\finiteautomaton_{1}^{\prime}) \cartesianproduct \automatonfinalstates(\finiteautomaton_{2}^{\prime})$.
  One can verify that $\automatonlang{\finiteautomaton^{\prime}} = \associatedlang{\relation_{1}\directsum \brelation}$.
  Therefore, $\relation_{1}\directsum \brelation$ is regular.
\end{proof}

Besides the operations described above, it is worth noting that, if $\relation_{1}$ and $\relation_{2}$ 
have the same arity, \ie $\arityrelation_1 = \arityrelation_2$, then the \emph{union} $\relation_{1} \cup \relation_{2}$ 
and the \emph{intersection} $\relation_{1} \cap \relation_{2}$ of $\relation_{1}$ and $\relation_{2}$ are regular relations. 
Moreover, if $\relation \subseteq (\alphabet^{+})^{{\cartesianproduct} \arityrelation}$ is a regular $\arityrelation$-ary relation, then the complement $\neg\relation \defeq (\alphabet^{+})^{{\cartesianproduct} \arityrelation} \setminus \relation$ of $\relation$ is also a regular $\arityrelation$-ary relation.

\subsection{Core Relations}

In this subsection we introduce some non-standard relations and prove that these relations are regular.
Intuitively, these relations will be used to implicitly manipulate tuples of ODDs, and in particular 
to construct a finite automaton accepting a string 
$\odd_0\tensorproduct \odd_1\tensorproduct \ldots \tensorproduct \odd_{\nrelations}\tensorproduct \vertex_1\ldots\vertex_{\nfreevariables}$ 
if and only if the tuple $(\odd_0,\odd_1,\ldots,\odd_{\nrelations})$ is $(\alphabet,\width,\vocabulary)$-structural 
and $\vertex_1,\ldots,\vertex_{\nfreevariables}$ belong to the domain $\oddlang{\odd_0}$.

\newcommand{\rr}[2]{\ensuremath{\mathcal{R}_{\in}(#1,#2)}}
\newcommand{\arr}[3]{\ensuremath{\widetilde{\mathcal{R}}(#1,#2,#3)}}
\newcommand{\crr}[3]{\ensuremath{\widetilde{\mathcal{R}}_{\in}(#1,#2,#3)}}
\newcommand{\nrr}[3]{\ensuremath{\widetilde{\mathcal{R}}_{\not\in}(#1,#2,#3)}}
\newcommand{\drr}[3]{\ensuremath{\mathcal{R}(#1,#2,#3)}}
\newcommand{\srr}[3]{\ensuremath{\mathcal{R}_{\subseteq}(#1,#2,#3)}}

\begin{proposition}\label{proposition:regularity_odd_alphabet}
  For each alphabet $\alphabet$ and each $\width \in \pN$, the language $\podddefiningset{\alphabet}{\width}$ is regular.
\end{proposition}
\begin{proof}
  Consider the finite automaton $\finiteautomaton$ over the alphabet $\layeralphabet{\alphabet}{\width}$ with state set $\automatonstates(\finiteautomaton) = \set*{\state_{\automatoninitialstates}} \cup \set{\state_{\layer} \setst \layer \in \layeralphabet{\alphabet}{\width}}$, transition set
  \begin{equation*}
      \begin{aligned}[t] 
        \automatontransitions(\finiteautomaton) = & \set*{\tuple{\state_{\automatoninitialstates}, \layer, \state_{\layer}} \setst \layer \in \layeralphabet{\alphabet}{\width}, \layerinitialflag(\layer) = \true} \\ 
        \cup & \set{\tuple{\state_{\blayer}, \alayer, \state_{\alayer}} \setst \alayer, \blayer \in \layeralphabet{\alphabet}{\width},\, \layerleftfrontier(\alayer) = \layerrightfrontier(\blayer), \layerfinalflag(\blayer) = \false, \layerinitialflag(\alayer) = \false}\text{,}
      \end{aligned}
  \end{equation*}
  initial state set $\automatoninitialstates(\finiteautomaton) = \set{\state_{\automatoninitialstates}}$ and final state set $\automatonfinalstates(\finiteautomaton) = \set*{\rchanged{\state_{\layer} \in \automatonstates(\finiteautomaton)} \setst \layerfinalflag(\layer) = \true}$.
  One can verify that $\automatonlang{\finiteautomaton} = \podddefiningset{\alphabet}{\width}$.
  Therefore, $\podddefiningset{\alphabet}{\width}$ is regular. 
\end{proof}

Let $\alphabet$ be an alphabet and $\width \in \pN$.
We let $\rr{\alphabet}{\width}$ be the relation defined as follows:
$$\rr{\alphabet}{\width} \defeq \big\{\tuple{\odd, \astring} \setst \odd \in \podddefiningset{\alphabet}{\width}, \astring \in \oddlang{\odd}\big\}\text{.}$$

\begin{proposition}\label{proposition:fundamental_relation}
  For each alphabet $\alphabet$ and each $\width \in \pN$, the relation $\rr{\alphabet}{\width}$ is regular.
\end{proposition}
\begin{proof}
  Consider the finite automaton $\finiteautomaton$ over the alphabet $\layeralphabet{\alphabet}{\width} \tensorproduct (\alphabet \dcup \set{\paddingsymbol})$ defined as follows.
  \begin{itemize}\setlength\itemsep{1.0ex}
    \item $\automatonstates(\finiteautomaton) = \set*{\state_{\automatoninitialstates}} \cup \set{\state_{\layer,{[\layerleftstate,\padasymbol,\layerrightstate]}} \setst \layer \in \layeralphabet{\alphabet}{\width}, \tuple{\layerleftstate,\padasymbol,\layerrightstate} \in \layertransitions(\layer)}$;
    
    \item 
    $\begin{aligned}[t] 
      \automatontransitions(\finiteautomaton) = & \set*{\tuple{{\state_{\automatoninitialstates}, \tuple{\layer,\padasymbol}, \state_{\layer,{[\layerleftstate,\padasymbol,\layerrightstate]}}}} \setst \state_{\layer,{[\layerleftstate,\padasymbol,\layerrightstate]}} \in \automatonstates(\finiteautomaton) \setminus \set{\state_{\automatoninitialstates}}, \layerleftstate \in \layerinitialstates(\layer)} \\ 
      \cup & \begin{multlined}[t][0.45\displaywidth] \Big\{\tuple{\state_{\blayer,{[\layerleftstate^{\prime},\padbsymbol,\layerrightstate^{\prime}]}}, \tuple{\alayer,\padasymbol}, \state_{\alayer, {[\layerleftstate,\padasymbol,\layerrightstate]}}} \setst \state_{\blayer,{[\layerleftstate^{\prime},\padbsymbol,\layerrightstate^{\prime}]}},\, \state_{\alayer,{[\layerleftstate,\padasymbol,\layerrightstate]}} \in \automatonstates(\finiteautomaton) \setminus \set{\state_{\automatoninitialstates}}, \\ \layerleftfrontier(\alayer) = \layerrightfrontier(\blayer), \layerfinalflag(\blayer) = \false, \layerinitialflag(\alayer) = \false, \layerleftstate = \layerrightstate^{\prime}\Big\}\text{;}\end{multlined}\end{aligned}$
    
    \item $\automatoninitialstates(\finiteautomaton) = \set{\state_{\automatoninitialstates}}$; 
    $\automatonfinalstates(\finiteautomaton) = \set*{\state_{\layer,{[\layerleftstate,\padasymbol,\layerrightstate]}} \in \automatonstates(\finiteautomaton) \setminus \set{\state_{\automatoninitialstates}} \setst  \layerrightstate \in \layerfinalstates(\layer)}$.
  \end{itemize}
  One can verify that $\automatonlang{\finiteautomaton} = \set*{\odd \tensorproduct \astring \setst \odd \in \podddefiningset{\alphabet}{\width}, \astring \in \oddlang{\odd}}$.
  In other words, $\associatedrelation{\automatonlang{\finiteautomaton}} = \rr{\alphabet}{\width}$.
  Therefore, $\rr{\alphabet}{\width}$ is a regular relation.
\end{proof}

\begin{lemma}\label{lemma:BasicAutomaticStructures}
 For each alphabet $\alphabet$ and each $\width, \anumber \in \pN$, the following relations are regular: 
 \begin{enumerate}
    \item \label{item:crr} $\begin{multlined}[t][0.9\displaywidth]\textstyle
    \crr{\alphabet}{\width}{\anumber} \defeq \big\{\tuple{\odd, \astring_{1},\ldots,\astring_{\anumber}} \setst \odd \in \podddefiningset{\alphabet}{\width},\, \astring_{1},\ldots,\astring_{\anumber} \in \oddlang{\odd}\big\}\text{;}\end{multlined}$
    
    \item \label{item:arr} $\begin{multlined}[t][0.9\displaywidth]\textstyle \arr{\alphabet}{\width}{\anumber} \defeq \set{\tuple{\odd, \astring_{1},\ldots,\astring_{\anumber}} \setst \odd \in \odddefiningset{\alphabet}{\width}{\oddlength},\, \astring_{1},\ldots,\astring_{\anumber} \in \alphabet^{\leq \oddlength}, \oddlength \in \pN}\text{;}\end{multlined}$

    \item \label{item:nrr} $\begin{multlined}[t][0.92\textwidth]\textstyle \nrr{\alphabet}{\width}{\anumber} \defeq \big\{\tuple*{\odd, \astring_{1},\ldots,\astring_{\anumber}} \setst \odd \in \odddefiningset{\alphabet}{\width}{\oddlength},\, \astring_{1},\ldots,\astring_{\anumber} \in \alphabet^{\leq \oddlength}, \\ \hfill \textstyle \astring_{i} \not\in \oddlang{\odd} \text{ for some } i \in \bset{\anumber},\, \oddlength \in \pN\big\}\text{.}\end{multlined}$
 \end{enumerate}
\end{lemma}
\begin{proof}
$ $ \\
  \ref{item:crr}. First, we prove that $\crr{\alphabet}{\width}{\anumber}$ is regular. 
  Consider the relation $\relation$ defined as follows:
$$\relation = \identify\big(\overbrace{\rr{\alphabet}{\width} \directsum \cdots \directsum \rr{\alphabet}{\width}}^{\anumber \text{ times }}, \set{\rchanged{(2i+1, 2i+3) \setst i \in \dbset{\anumber-1}}}\big)\text{.}$$
  In other words, $\relation$ consists of all tuples $\tuple*{\odd, \astring_{1}, \ldots, \odd, \astring_{\anumber}}$ such that $\odd \in \podddefiningset{\alphabet}{\width}$ and $\astring_{1}, \ldots, \astring_{\anumber} \in \oddlang{\odd}$.
  Therefore, $\crr{\alphabet}{\width}{\anumber}$ is regular, since it may be defined as the regular relation $\projection(\relation, \set{2i+1 \setst i \in \bset{\anumber - 1}})$.

  \ref{item:arr}. Consider the finite automaton $\finiteautomaton$ over the alphabet $\layeralphabet{\alphabet}{\width} \tensorproduct (\alphabet \dcup \set{\paddingsymbol})^{\tensorproduct \anumber}$ defined as follows:
  \begin{itemize}\setlength\itemsep{1.0ex}
    \item $\automatonstates(\finiteautomaton) = \set*{\state_{\automatoninitialstates}} \cup \set*{\state_{\layer,{[\padasymbol_{1},\ldots,\padasymbol_{\anumber}]}} \setst \layer \in \layeralphabet{\alphabet}{\width},\, \padasymbol_{1},\ldots,\padasymbol_{\anumber} \in \alphabet \dcup \set{\paddingsymbol}}$;
    
    \item 
    $\begin{aligned}[t] 
      \automatontransitions(\finiteautomaton) = & \begin{multlined}[t][0.85\textwidth]\Big\{\tuple*{\state_{\automatoninitialstates}, \tuple{\layer,\asymbol_{1},\ldots,\asymbol_{\anumber}}, \state_{\layer,[\asymbol_{1},\ldots,\asymbol_{\anumber}]}} \setst \layer \in \layeralphabet{\alphabet}{\width}, \layerinitialflag(\layer) = \true,\, \\[-1.5ex] \hfill \asymbol_{1},\ldots,\asymbol_{\anumber} \in \alphabet\Big\}\end{multlined} \\
      \cup & \begin{multlined}[t][0.85\textwidth] \Big\{\tuple{\state_{\blayer,[\padbsymbol_{1},\ldots,\padbsymbol_{\anumber}]}, \tuple{\alayer,\padasymbol_{1},\ldots,\padasymbol_{\anumber}}, \state_{\alayer, {[\padasymbol_{1},\ldots,\padasymbol_{\anumber}]}}} \setst \alayer, \blayer \in \layeralphabet{\alphabet}{\width}, \\[0.5ex] \hfill \layerleftfrontier(\alayer) = \layerrightfrontier(\blayer), \layerfinalflag(\blayer) = \false, \layerinitialflag(\alayer) = \false,\, \padbsymbol_{i},\padasymbol_{i} \in \alphabet \dcup \set{\paddingsymbol}, \hspace{4.0ex} \\[-1.0ex] \hfill \rchanged{(\padbsymbol_{i} = \paddingsymbol \Rightarrow \padasymbol_{i} = \paddingsymbol) \text{ for each } i \in \bset{\anumber}}\Big\}\text{;}\end{multlined}\end{aligned}$ 
	      
    \item $\automatoninitialstates(\finiteautomaton) = \set{\state_{\automatoninitialstates}}$; 
    $\automatonfinalstates(\finiteautomaton) = \set*{\state_{\layer,{[\padasymbol_{1},\ldots,\padasymbol_{\anumber}]}} \in \automatonstates(\finiteautomaton) \setminus \set{\state_{\automatoninitialstates}} \setst \layerfinalflag(\layer) = \true}$.
  \end{itemize}
  One can verify that $$\automatonlang{\finiteautomaton} = \set*{\odd \tensorproduct (\astring_{1}\tensorproduct\cdots\tensorproduct\astring_{\anumber}) \setst \odd \in \odddefiningset{\alphabet}{\width}{\oddlength},\, \astring_{1},\ldots,\astring_{\anumber} \in \alphabet^{\leq \oddlength}, \oddlength \in \pN}\text{.}$$
  Therefore, $\arr{\alphabet}{\width}{\anumber}$ is a regular relation.

  \ref{item:nrr}. Note that, $\nrr{\alphabet}{\width}{\anumber}$ consists of all tuples $\tuple*{\odd, \astring_{1}, \ldots, \astring_{\anumber}}\in \arr{\alphabet}{\width}{\oddlength}$ such that $\astring_{i} \not\in \oddlang{\odd}$ for some $i \in \bset{\anumber}$.
  Therefore, $\nrr{\alphabet}{\width}{\anumber}$ may be defined as the regular relation $\arr{\alphabet}{\width}{\anumber} \setminus \crr{\alphabet}{\width}{\anumber} = \arr{\alphabet}{\width}{\anumber} \cap \neg\crr{\alphabet}{\width}{\anumber}$.
\end{proof}

\begin{lemma}\label{lemma:drr}
  For each two alphabets $\aalphabet$ and $\balphabet$ and each $\width \in \pN$, the relation $\drr{\aalphabet}{\balphabet}{\width} \defeq \set*{\tuple*{\aodd, \bodd} \setst \aodd \in \podddefiningset{\aalphabet}{\width},\, \bodd \in \podddefiningset{\balphabet}{\width}}$ is regular.
\end{lemma}
\begin{proof}
  It follows from Proposition~\ref{proposition:regularity_odd_alphabet} that $\podddefiningset{\aalphabet}{\width}$ and $\podddefiningset{\balphabet}{\width}$ are regular languages.
  In other words, $\relation_{1} = \podddefiningset{\aalphabet}{\width}$ and $\relation_{2} = \podddefiningset{\balphabet}{\width}$ are regular unary relations.
  Therefore, $\drr{\aalphabet}{\balphabet}{\width}$ is regular, since it may be defined as the regular relation $\relation_{1} \directsum \relation_{2}$.
\end{proof}

Let $\alphabet$ be an alphabet and $\width, \arityrelation \in \pN$.
We let $\srr{\alphabet}{\width}{\arityrelation}$ be the relation defined as follows:
\begin{equation*}
    \srr{\alphabet}{\width}{\arityrelation} \defeq \big\{\tuple*{\aodd, \bodd} \setst \aodd \in \podddefiningset{\alphabet}{\width},\, \bodd \in \podddefiningset{\alphabet^{\tensorproduct \arityrelation}}{\width},\, \oddlang{\bodd} \subseteq \oddlang{\aodd}^{\tensorproduct \arityrelation}\big\}\text{.}
\end{equation*}

\begin{proposition}\label{proposition:srr}
  For each alphabet $\alphabet$ and each $\width, \arityrelation \in \pN$, the relation $\srr{\alphabet}{\width}{\arityrelation}$ is regular.
\end{proposition}
\begin{proof}
  Consider the relation $\arelation = \fold(\nrr{\alphabet}{\width}{\arityrelation},2,\arityrelation+1) \directsum \rr{\alphabet^{\tensorproduct \arityrelation}}{\width}$.
  Note that, $\arelation$ consists of all tuples of the form $\tuple*{\aodd, \astring_{1} \tensorproduct \cdots \tensorproduct \astring_{\arityrelation}, \bodd, \bstring_{1} \tensorproduct \cdots \tensorproduct \bstring_{\arityrelation}}$ satisfying the following conditions: 
  \begin{itemize}
    \item $\aodd \in \odddefiningset{\alphabet}{\width}{\oddlength}$, $\astring_{1}, \ldots, \astring_{\arityrelation} \in \alphabet^{\leq \oddlength}$ but, for some $i \in \bset{\arityrelation}$, $\astring_{i} \not\in \oddlang{\aodd}$;
    \item $\bodd \in \odddefiningset{\alphabet}{\width}{\oddlength^{\prime}}$, and $\bstring_{1} \tensorproduct \cdots \tensorproduct \bstring_{\arityrelation} \in \oddlang{\bodd}$, 
  \end{itemize}
  where $\oddlength, \oddlength^{\prime} \in \pN$. 
  Let $\brelation = \identify(\arelation, 2, 4)$. 
  By definition, $\brelation$ is the sub-relation of $\arelation$ comprising all tuples $\tuple{\aodd,\astring_{1} \tensorproduct \cdots \tensorproduct \astring_{\arityrelation}, \bodd, \bstring_{1} \tensorproduct \cdots \tensorproduct \bstring_{\arityrelation}} \in \arelation$ such that $\astring_{i} = \bstring_{i}$ for each $i \in \bset{\arityrelation}$.
  Thus, $\projection(\brelation, \set{2, 4})$ consists of all tuples $\tuple*{\aodd,\bodd} \in \odddefiningset{\alphabet}{\width}{\oddlength} \cartesianproduct \odddefiningset{\alphabet^{\tensorproduct \arityrelation}}{\width}{\oddlength^{\prime}}$ such that there exist $\astring_{1}, \ldots, \astring_{\arityrelation} \in \alphabet^{\leq \oddlength}$ with $\astring_{1} \tensorproduct \cdots \tensorproduct \astring_{\arityrelation} \in \oddlang{\bodd}$ but, for some $i \in \bset{\arityrelation}$, $\astring_{i} \not\in \oddlang{\aodd}$, where $\oddlength, \oddlength^{\prime} \in \pN$.
  In other words, we have that
  \begin{equation*}
      \projection(\brelation, \set{2, 4}) = \big\{\tuple*{\aodd,\bodd} \setst \aodd \in \podddefiningset{\alphabet}{\width}, \bodd \in \podddefiningset{\alphabet^{\tensorproduct \arityrelation}}{\width}, \oddlang{\bodd} \nsubseteq \oddlang{\aodd}^{\tensorproduct \arityrelation}\big\}\text{,}
  \end{equation*}
  Now, let $\drr{\alphabet}{\alphabet^{\tensorproduct \arityrelation}}{\width} = \set{\tuple{\aodd, \bodd} \setst \aodd \in \podddefiningset{\alphabet}{\width},\, \bodd \in \podddefiningset{\alphabet^{\tensorproduct \arityrelation}}{\width}}$.
  From Lemma~\ref{lemma:drr}, $\drr{\alphabet}{\alphabet^{\tensorproduct \arityrelation}}{\width}$ is regular.
  Therefore, $\srr{\alphabet}{\width}{\arityrelation}$ is regular, since it may be simply regarded as the relation $\drr{\alphabet}{\alphabet^{\tensorproduct \arityrelation}}{\width} \cap \neg\projection(\brelation, \set{2, 4})$.
\end{proof}

\subsection{Regular-Structural Relations}

\newcommand{\frr}[4]{\ensuremath{\mathcal{R}(#1,#2,#3,#4)}}

Let $\alphabet$ be an alphabet, $\width \in \pN$, $\vocabulary = \tuple{\relationsymbol_{1},\ldots,\relationsymbol_{\nrelations}}$ be a relational vocabulary and $\nfreevariables \in \N$.
We let $\frr{\alphabet}{\width}{\vocabulary}{\nfreevariables}$ be the relation defined as follows: 
\begin{equation*}
  \begin{multlined}[0.9\displaywidth]\textstyle
    \frr{\alphabet}{\width}{\vocabulary}{\nfreevariables} \defeq \big\{\tuple{\odd_{0}, \odd_{1}, \ldots, \odd_{\nrelations}, \vertex_{1}, \ldots, \vertex_{\nfreevariables}} \setst \\ \tuple*{\odd_{0}, \odd_{1}, \ldots, \odd_{\nrelations}} \in \derivedrelation(\alphabet,\width,\vocabulary),\, \vertex_{1},\ldots,\vertex_{\nfreevariables} \in \oddlang{\odd_{0}}\big\}\text{.}
  \end{multlined}
\end{equation*}
In particular, note that, if $\nfreevariables = 0$, then $\frr{\alphabet}{\width}{\vocabulary}{\nfreevariables} = \derivedrelation(\alphabet,\width,\vocabulary)$.

\begin{lemma}\label{lemma:frr_regular}
 Let $\alphabet$ be an alphabet, $\width \in \pN$, $\vocabulary = \tuple{\relationsymbol_{1}, \ldots, \relationsymbol_{\nrelations}}$ be a relational vocabulary and $\nfreevariables \in \N$.
  The relation $\frr{\alphabet}{\width}{\vocabulary}{\nfreevariables}$ is regular.
\end{lemma}
\begin{proof}
  Consider the relations $\arelation = \srr{\alphabet}{\width}{\arityrelation_{1}} \directsum \cdots \directsum \srr{\alphabet}{\width}{\arityrelation_{\nrelations}} \directsum \crr{\alphabet}{\width}{\nfreevariables}$ 
and $\brelation = \identify(\arelation,\set{\tuple{2i+1, 2i+3} \setst i \in \dbset{\nrelations}})$.
  One can verify that $\brelation$ consists of tuples of the form $\tuple{\odd_{0}, \odd_{1}, \ldots, \odd_{0}, \odd_{\nrelations}, \odd_{0}, \vertex_{1}, \ldots, \vertex_{\nfreevariables}}$ satisfying the conditions: 
$\tuple*{\odd_{0}, \odd_{1}, \ldots, \odd_{\nrelations}} \in \derivedrelation(\alphabet,\width,\vocabulary)$ 
  and $\vertex_{1},\ldots,\vertex_{\nfreevariables} \in \oddlang{\odd_{0}}$.
  Therefore,  $\frr{\alphabet}{\width}{\vocabulary}{\nfreevariables}$ is regular, since it may regarded as the relation $\projection(\brelation,\set{2i+1 \setst i \in \bset{\nrelations}})$.
\end{proof}

Let $\alphabet$ be a finite alphabet, $\width\in \pN$, $\vocabulary$ be a relational vocabulary, and 
$\formula(\variable_{1},\ldots,\variable_{\nfreevariables})$ be an $\fol{\vocabulary}$-formula, with 
free variables $\variable_1,\dots,\variable_{\nfreevariables}$. Consider the following relation. 
\begin{equation*}
\begin{multlined}[0.9\displaywidth]\textstyle
	\derivedrelation(\alphabet,\width,\vocabulary,\formula(\avariable_1,\dots,\avariable_{\nfreevariables})) \defeq \big\{\tuple{\odd_{0}, \odd_{1}, \ldots, \odd_{\nrelations}, \vertex_{1}, \ldots, \vertex_{\nfreevariables}} \in \frr{\alphabet}{\width}{\vocabulary}{\nfreevariables} \setst \\ \derivedstructure(\odd_{0}, \odd_{1}, \ldots, \odd_{\nrelations}) \models \formula[\vertex_{1}, \ldots, \vertex_{\nfreevariables}]\}\text{.}
  \end{multlined}
\end{equation*}

In Theorem \ref{theorem:RegularInduction}, we show that
$\derivedrelation(\alphabet,\width,\vocabulary,\formula(\avariable_1,\dots,\avariable_{\nfreevariables}))$
is a regular relation. We remark that if $\nfreevariables = 0$, then $\formula$ is an $\fol{\vocabulary}$-sentence and, 
in this case, the relation  $\derivedrelation(\alphabet,\width,\vocabulary,\formula(\variable_1,\dots,\variable_{\nfreevariables}))$
coincides with the relation $\derivedrelation(\alphabet,\width,\vocabulary,\formula)$. Therefore, 
Theorem \ref{theorem:RegularInduction} implies Theorem~\ref{theorem:FirstOrderRegular}. 

The proof of Theorem~\ref{theorem:RegularInduction} is by induction on the structure of the input formula
$\formula(\variable_1,\ldots,\variable_{\nfreevariables})$. The base case follows by combining the auxiliary 
lemmas and propositions proven in this section. The induction step follows 
from the fact that regular languages are closed under, negation, union, intersection and projection.

\begin{theorem}
\label{theorem:RegularInduction}
Let $\alphabet$ be an alphabet, $\width \in \pN$, $\vocabulary = \tuple{\relationsymbol_{1}, \ldots, \relationsymbol_{\nrelations}}$ 
be a relational vocabulary and $\formula(\variable_1,\ldots,\variable_{\nfreevariables})$ be an
$\fol{\vocabulary}$-formula. The relation $\derivedrelation(\alphabet,\width,\vocabulary,\formula(\avariable_1,\dots,\avariable_{\nfreevariables}))$
is regular. 
\end{theorem}
\begin{proof}
We prove by induction on the structure of the formula $\formula(\variable_1,\ldots,\variable_{\nfreevariables})$
that the relation $\derivedrelation(\alphabet,\width,\vocabulary,\formula(\avariable_1,\dots,\avariable_{\nfreevariables}))$
is regular.

  \textit{Base case.} Suppose that $\formula(\avariable_1,\ldots,\avariable_{\nfreevariables})$ is an atomic formula.
 Then, there are two subcases to be analyzed.
 In the first subcase, $\formula(\variable_1,\ldots,\variable_{\nfreevariables}) \equiv (\variable_{i} = \variable_{j})$
for some $i,j \in \bset{\nfreevariables}$.
In this case, one can verify that
$\derivedrelation(\alphabet,\width,\vocabulary,\formula,\freevariableset_{\nfreevariables})$
is equal to the relation 
$\identify(\rchanged{\frr{\alphabet}{\width}{\vocabulary}{\nfreevariables}}, \rchanged{\nrelations + i+1, \nrelations  + j+1})$ 
which consists of the subset of tuples in $\frr{\alphabet}{\width}{\vocabulary}{\nfreevariables}$ where the 
coordinate corresponding to variable $\avariable_i$ is equal to the coordinate corresponding to variable $\avariable_j$.
We note that we need to add $\nrelations+1$ to $i$ and $j$ because the first 
$\nrelations+1$ coordinates of each tuple in $\frr{\alphabet}{\width}{\vocabulary}{\nfreevariables}$ are ODDs, while 
the last $\nfreevariables$ coordinates are strings representing assignments to the variables. 

In the second subcase, $\formula(\variable_1,\ldots,\variable_{\nfreevariables}) \equiv \relationsymbol_{i}(\avariable_{\cmap(1)}, \ldots, \avariable_{\cmap(\arityrelation_{i})})$ for some $\cmap \colon \bset{\arityrelation_{i}} \rightarrow \bset{\nfreevariables}$ and some $i \in \bset{\nrelations}$.
  Let $\relation = \unfold(\rr{\alphabet^{\tensorproduct \arityrelation_{i}}}{\width},2)$. Intuitively, $\relation$ consists of all tuples of the form $\tuple{\odd, \bvertex_{1}, \ldots, \bvertex_{\arityrelation_{i}}}$ such that $\odd \in \odddefiningset{\alphabet^{\tensorproduct \arityrelation_{i}}}{\width}{\oddlength}$ and $\bvertex_{1} \tensorproduct \cdots \tensorproduct \bvertex_{\arityrelation_{i}} \in \oddlang{\odd}$.
  Let $\relation^{\prime} = \identify(\relation \directsum \rchanged{\frr{\alphabet}{\width}{\vocabulary}{\nfreevariables}}, 1, i+\arityrelation_{i}+2)$.
  One can verify that $\relation^{\prime}$ comprises all tuples of the form $\tuple{\odd_{i}, \bvertex_{1}, \ldots, \bvertex_{\arityrelation_{i}}, \odd_{0}, \odd_{1}, \ldots, \odd_{i}, \ldots, \odd_{\nrelations}, \avertex_{1}, \ldots, \avertex_{\nfreevariables}}$ such that $\tuple{\odd_{0}, \odd_{1}, \ldots, \odd_{\nrelations}, \avertex_{1}, \ldots, \avertex_{\nfreevariables}} \in \rchanged{\frr{\alphabet}{\width}{\vocabulary}{\nfreevariables}}$ and $\bvertex_{1} \tensorproduct \cdots \tensorproduct \bvertex_{\arityrelation_{i}} \in \oddlang{\odd_{i}}$.
  Let $\relation^{\prime\prime} = \identify(\relation^{\prime}, \set{\tuple{j+1, \cmap(j)+\arityrelation_{i}+\nrelations+2} \setst j \in \bset{\arityrelation_{i}}})$.
  In other words, $\relation^{\prime\prime}$ consists of all tuples $\tuple{\odd_{i}, \bvertex_{1}, \ldots, \bvertex_{\arityrelation_{i}}, \odd_{0}, \odd_{1}, \ldots, \odd_{i}, \ldots, \odd_{\nrelations}, \avertex_{1}, \ldots, \avertex_{\nfreevariables}} \in \relation^{\prime}$ such that $\bvertex_{j} = \avertex_{\cmap(j)}$ for each $j \in \bset{\arityrelation_{i}}$, which implies $\avertex_{\cmap(1)} \tensorproduct \cdots \tensorproduct \avertex_{\cmap(\arityrelation_{i})} \in \oddlang{\odd_{i}}$.
  As a result, we obtain the regularity of $\derivedrelation(\alphabet,\width,\vocabulary,\formula,\freevariableset_{\nfreevariables})$, since in this case $\derivedrelation(\alphabet,\width,\vocabulary,\formula,\freevariableset_{\nfreevariables})$ 
can be defined as the relation $\projection(\relation^{\prime\prime}, \set{1, \ldots, \arityrelation_{i}+1})$.

  \textit{Inductive step.} It suffices to analyze the following four subcases: $\formula \equiv \formula_{1} \lor \formula_{2}$, $\formula \equiv \formula_{1} \land \formula_{2}$, $\formula \equiv \neg\formula_{1}$ and $\formula(\variable_1,\ldots,\variable_{\nfreevariables}) \equiv \exists\,\bvariable\; \formula_{3}(\avariable_{1},\ldots,\avariable_{\nfreevariables},\bvariable)$ for 
$\fol{\vocabulary}$-formulas $\formula_1$, $\formula_2$ and $\formula_3$ and variable $\bvariable$ 
that is free in $\formula_3$. 

First, assume that $\formula \equiv \formula_{1} \lor \formula_{2}$.
By the induction hypothesis, 
$\derivedrelation(\alphabet,\width,\vocabulary,\formula_{1}(\avariable_1,\dots,\avariable_{\nfreevariables}))$
and $\derivedrelation(\alphabet,\width,\vocabulary,\formula_{2},\formula(\avariable_1,\dots,\avariable_{\nfreevariables}))$
are regular relations. This implies that
$\derivedrelation(\alphabet,\width,\vocabulary,\formula(\avariable_1,\dots,\avariable_{\nfreevariables}))$
is a regular relation, since it is equal to the union
$$\derivedrelation(\alphabet,\width,\vocabulary,\formula_{1}(\avariable_1,\dots,\avariable_{\nfreevariables}))\cup
	\derivedrelation(\alphabet,\width,\vocabulary,\formula_{2}(\avariable_1,\dots,\avariable_{\nfreevariables}))$$
of two regular relations. Similarly, assume that $\formula \equiv \formula_{1} \land \formula_{2}$. 
By the induction hypothesis, the relations 
$\derivedrelation(\alphabet,\width,\vocabulary,\formula_{1}(\avariable_1,\dots,\avariable_{\nfreevariables}))$
and $\derivedrelation(\alphabet,\width,\vocabulary,\formula_{2}(\avariable_1,\dots,\avariable_{\nfreevariables}))$
are regular. Then
$\derivedrelation(\alphabet,\width,\vocabulary,\formula(\avariable_1,\dots,\avariable_{\nfreevariables})$
is a regular relation, since it is equal to the intersection 
$$\derivedrelation(\alphabet,\width,\vocabulary,\formula_{1}(\avariable_1,\dots,\avariable_{\nfreevariables}))
\cap
 \derivedrelation(\alphabet,\width,\vocabulary,\formula_{2}(\avariable_1,\dots,\avariable_{\nfreevariables}))$$
of two regular relations. 

  Now, assume that $\formula \equiv \neg\formula_{1}$.
By the induction hypothesis, we have that the relation 
$\derivedrelation(\alphabet,\width,\vocabulary,\formula_{1}(\avariable_1,\dots,\avariable_{\nfreevariables}))$
is regular.
Therefore, $\derivedrelation(\alphabet,\width,\vocabulary,\formula(\avariable_1,\dots,\avariable_{\nfreevariables}))$
is regular, since it is equal to the Boolean combination
$\neg \derivedrelation(\alphabet,\width,\vocabulary,\formula_{1}(\avariable_1,\dots,\avariable_{\nfreevariables})) 
\cap \rchanged{\frr{\alphabet}{\width}{\vocabulary}{\nfreevariables}}$ of two regular relations. 

Finally, assume that
$\formula(\avariable_{1},\ldots,\avariable_{\nfreevariables}) \equiv \exists\bvariable\;
\formula_{3}(\avariable_{1},\ldots,\avariable_{\nfreevariables},\bvariable)$ for some variable
$\bvariable$ that is free in $\formula_{3}$. By the induction hypothesis,
$\derivedrelation(\alphabet,\width,\vocabulary,\formula_{3}(\avariable_1,\dots,\avariable_{\nfreevariables},\bvariable))$
is a regular relation.
Therefore, $\derivedrelation(\alphabet,\width,\vocabulary,\formula(\avariable_1,\dots,\avariable_{\nfreevariables}))$
is regular, since it can be obtained from 
$\derivedrelation(\alphabet,\width,\vocabulary,\formula_{3}(\avariable_1,\dots,\avariable_{\nfreevariables},\bvariable))$
by deleting the coordinate corresponding to $\bvariable$ from each tuple in 
$\derivedrelation(\alphabet,\width,\vocabulary,\formula_{3}(\avariable_1,\dots,\avariable_{\nfreevariables},\bvariable))$.
	In other words, $\derivedrelation(\alphabet,\width,\vocabulary,\formula,\freevariableset_{\nfreevariables}) = \projection(\derivedrelation(\alphabet,\width,\vocabulary,\formula_{3}(\avariable_1,\dots,\avariable_{\nfreevariables},\bvariable)),
\nrelations + \nfreevariables + 2)$.
\end{proof}

\section{Counting Satisfying Assignments} 
\label{section:CountingSatisfyingAssignments}

Let $\vocabulary=(\relationsymbol_1,\ldots,\relationsymbol_{\nrelations})$ be a relational vocabulary, 
$\formula(\variable_1,\ldots,\variable_{\nfreevariables})$ be an $\fol{\vocabulary}$-formula 
with $\freevariables(\formula) \subseteq \set{\variable_1,\ldots,\variable_{\nfreevariables}}$, 
and let $\astructure$ be a finite $\vocabulary$-structure. 
We say that an assignment $(\vertex_1,\ldots,\vertex_{\nfreevariables})\in \relationsymbol_0(\structure)^{\nfreevariables}$
of the variables $(\variable_1,\ldots,\variable_{\nfreevariables})$ \emph{satisfies $\formula$ with respect to $\structure$}
if $\structure \models \formula[\vertex_1,\ldots,\vertex_{\nfreevariables}]$. 

The next theorem states that given a $(\alphabet,\width,\vocabulary)$-structural tuple $\structuraltuple$ such that 
$\derivedstructure(\structuraltuple) = \structure$, the problem of counting the number of satisfying assignments
of $\formula$ with respect to $\structure$ 
is fixed parameter tractable when parameterized by $\alphabet$, $\width$, $\vocabulary$, $\formula$ and $\nfreevariables$. 
More specifically, such a problem can be solved in time quadratic in the length of the $(\alphabet,\width,\vocabulary)$-structural tuple $\structuraltuple$. 
We remark that, in most applications, the parameters $\alphabet,\vocabulary,\formula,\nfreevariables$ are naturally already fixed, and therefore the only complexity parameter that is indeed relevant in these situations is the width $\width$. 

\begin{theorem}\label{theorem:solutions}
Let $\alphabet$ be an alphabet, $\width \in \pN$, $\vocabulary = \tuple{\relationsymbol_{1}, \ldots, \relationsymbol_{\nrelations}}$ be a 
relational vocabulary and $\formula(\variable_{1},\ldots,\variable_{\nfreevariables})$ be an $\fol{\vocabulary}$-formula. 
Given a $(\alphabet,\width,\vocabulary)$-structural tuple $\structuraltuple = (\odd_0,\odd_1\ldots\odd_{\nrelations})$ 
of length $\oddlength$, one can count in time $f(\alphabet,\width,\vocabulary,\formula,\nfreevariables) \cdot \oddlength^2$, 
for some computable function $f$, the number of satisfying assignments for $\formula$ with respect to $\derivedstructure(\structuraltuple)$. 
\end{theorem}
\begin{proof}
     \rchanged{Let $\freevariableset_{\nfreevariables} = \set{\variable_1,\ldots,\variable_{\nfreevariables}}$ and $\finiteautomaton_{1}$ be the minimum deterministic finite automaton with language $\automatonlang{\finiteautomaton_{1}} =
     \associatedlang{\derivedrelation(\alphabet,\width,\vocabulary,\formula,\freevariableset_{\nfreevariables})}$. 
     
     Assume that $\odd_{0} =
     \layer_{1}^{0}\cdots\layer_{\oddlength}^{0} \in
     \odddefiningset{\alphabet}{\width}{\oddlength}$, and 
     for each $i\in [\oddlength]$, $\odd_{i} = \layer_{1}^{i}\cdots\layer_{\oddlength}^{i} \in
     \odddefiningset{\alphabet^{\tensorproduct \arityrelation_{i}}}{\width}{\oddlength}$, where $\arityrelation_i$ denotes 
     the arity of $\relationsymbol_i$.  
     Then, consider} the finite automaton $\finiteautomaton_{2}$ defined as follows: \begin{itemize}
   \item $\automatonstates(\finiteautomaton_{2}) = \set{\state_{0}} \cup
     \set{\state_{i,
     [\padasymbol_{1},\ldots,\padasymbol_{\nfreevariables}]} \setst
     \padasymbol_{1}, \ldots, \padasymbol_{\nfreevariables} \in
     \alphabet \dcup \set{\paddingsymbol},\, i \in
     \bset{\oddlength}}\text{;}$\vspace{1.5ex}

    \item
      $\begin{aligned}[t]\automatontransitions(\finiteautomaton_{2})
        & = \big\{\tuple{\state_{0},
        \tuple{\layer_{1}^{0}, \layer_{1}^{1}, \ldots,
        \layer_{1}^{\nrelations}, \asymbol_{1}, \ldots,
        \asymbol_{\nfreevariables}}, \state_{1,
        [\asymbol_{1}, \ldots,
        \asymbol_{\nfreevariables}]}} \setst
        \asymbol_{1}, \ldots,
        \asymbol_{\nfreevariables} \in \alphabet\big\}
        \\[-0.5ex] & \cup
        \begin{multlined}[t][0.835\textwidth]\Big\{\tuple{\state_{i,
        [\padbsymbol_{1}, \ldots,
        \padbsymbol_{\nfreevariables}]},
        \tuple{\layer_{i+1}^{0}, \layer_{i+1}^{1},
        \ldots, \layer_{i+1}^{\nrelations},
        \padasymbol_{1}, \ldots,
        \padasymbol_{\nfreevariables}}, \state_{i+1,
        [\padasymbol_{1}, \ldots,
        \padasymbol_{\nfreevariables}]}} \setst
        \\[0.2ex] \hfill \rchanged{\padbsymbol_{j}, \padasymbol_{j} \in \alphabet\dcup\set{\paddingsymbol}},\, (\padbsymbol_{j} =
        \paddingsymbol \Rightarrow \padasymbol_{j} =
        \paddingsymbol) \text{ for each } j \in
        \bset{\nfreevariables}, \hspace{4.0ex}
        \\[-1.0ex] \hfill i \in
        \bset{\oddlength-1}\Big\}\text{;}\end{multlined}\end{aligned}$\vspace{1.5ex}

    \item $\automatoninitialstates(\finiteautomaton_{2}) =
      \set{\state_{0}}$;
      \rchanged{$\automatonfinalstates(\finiteautomaton_{2}) =
            \set{\state_{\oddlength, [\padasymbol_{1}, \ldots,
            \padasymbol_{\nfreevariables}]} \setst \padasymbol_{1}, \ldots, \padasymbol_{\nfreevariables} \in
           \alphabet \dcup \set{\paddingsymbol}}$}.
      \end{itemize} 
    
     \rchanged{One can verify that} the language $\automatonlang{\finiteautomaton_2}$ consists of all strings of the form
     
     $$\odd_0\tensorproduct \odd_1\tensorproduct \cdots \tensorproduct \odd_{\nrelations} \tensorproduct 
      \vertex_1\tensorproduct \cdots \tensorproduct \vertex_{\nfreevariables}$$
     such that, for each $i\in \bset{\nfreevariables}$, $\vertex_i$ is a string in $\alphabet^{\leq \stringlength}$.

      Note that, the set of states of $\finiteautomaton_2$ \rchanged{can be} split into $\oddlength$ levels \rchanged{such that} states in level $i$ only have transitions to states in level $i+1$. Therefore, $\finiteautomaton_2$ is acyclic. Additionally, one can readily check that $\finiteautomaton_2$
      is deterministic and \rchanged{that} the number of states in each level of \rchanged{$\finiteautomaton_2$ can be described as} a function of $\abs{\alphabet}$ and $\nfreevariables$. 

      Now, consider the automaton $\finiteautomaton_{\cap} =  \finiteautomaton_{1} \cap \finiteautomaton_{2}$ 
      whose language is the intersection of the languages of $\finiteautomaton_1$ and of $\finiteautomaton_2$.
      Note that $\finiteautomaton_{\cap}$ accepts precisely the strings of the form $\odd_0\tensorproduct \odd_1\tensorproduct \cdots \tensorproduct \odd_{\nrelations} \tensorproduct 
            \vertex_1\tensorproduct\cdots \tensorproduct \vertex_{\nfreevariables}$
      such that $(\vertex_1\ldots\vertex_{\nfreevariables})$ is a satisfying assignment for
      $\formula$ with respect to the $\vocabulary$-structure $\derivedstructure(\structuraltuple)$.
      Furthermore, \rchanged{one can verify that $\finiteautomaton_{\cap}$} is also \rchanged{an acyclic deterministic automaton whose states can be split into $\oddlength$ levels such that} there are only transitions going from level $i$ to level $i+1$.
      
      Therefore, counting the number of assignments \rchanged{that satisfy
      $\formula$ with respect to $\derivedstructure(\structuraltuple)$} amounts
      to counting the number of accepting paths in $\finiteautomaton_{\cap}$.
      This can be done in time $\BigOh(\abs{\finiteautomaton_{\cap}}^2)$ 
      using standard dynamic programming, starting from
      \rchanged{the states of $\finiteautomaton_{\cap}$ that are in the last
      level and, then, recording, for each state of each level, the number of
      paths from that state to a final state.} \rchanged{Since
      $\finiteautomaton_{\cap}$ can be constructed in time
      $g(\alphabet,\width,\vocabulary,\formula,\nfreevariables)\cdot \oddlength$, for some computable function $g$,} 
      the overall algorithm takes time $\BigOh(g(\alphabet,\width,\vocabulary,\formula,\nfreevariables)^2 \cdot \oddlength^2)$. By setting 
      $f(\alphabet,\width,\vocabulary,\formula,\nfreevariables) = c\cdot g(\alphabet,\width,\vocabulary,\formula,\nfreevariables)^2+d$ for 
      some sufficiently large constants $c$ and $d$, the theorem follows.\end{proof}

\section{Conclusion}

In this work, we have introduced the notion of decisional-width of a relational structure, 
as a measure that intuitively tries to provide a quantification of how difficult it is to 
define the relations of the structure. Subsequently we provided a suitable way of defining
infinite classes of structures of small width. Interestingly, there exist classes of structures 
of constant decisional-width that have very high width with respect to most width measures defined 
so far in structural graph theory. As an example, we have shown that the class of hypercube graphs has
regular-decisional width $2$, while it is well known that they have unbounded treewidth and cliquewidth. Additionally, 
this family is not nowhere-dense. Therefore, first-order model-checking and validity-testing techniques 
developed for these well studied classes of graphs do not generalize to graphs of constant decisional
width. Other examples of families of graphs of constant decisional width 
are paths, cliques (which have unbounded treewidth), unlabeled grids (which have unbounded treewidth and cliquewidth)
and many others. It is interesting to note that these mentioned classes have all a regular structure
and, therefore, are ``easy'' to describe.

\subsubsection*{Acknowledgements} 

Alexsander Andrade de Melo acknowledges support from the Brazilian National Council for Scientific and Technological Development (CNPQ 140399/2017-8)
and from the Brazilian Federal Agency for Support and Evaluation of Graduate Education (CAPES 88881.187636/2018-01). 
Mateus de Oliveira Oliveira acknowledges support from the Bergen Research Foundation, from the Research Council of Norway (288761) and from 
the Sigma2 network (NN9535K).

\bibliographystyle{spmpsci}      
\bibliography{implicitGraphs}   

\end{document}